\documentclass[12pt]{article}
\usepackage[utf8]{inputenc}

\usepackage{amsmath,amsthm,verbatim,amssymb,amsfonts,amscd, graphicx}
\usepackage{graphics}
\usepackage{mathtools}
\usepackage{comment}
\usepackage[all]{xy}
\usepackage{comment}
\usepackage[dvipsnames]{xcolor}
\usepackage{hyperref}
\usepackage[capitalise, noabbrev]{cleveref}
\usepackage{caption}
\usepackage{subcaption}
\usepackage{authblk}

\newcommand{\R}{\mathbb{R}}

\newcommand{\Pt}{P_t}
\newcommand{\Ps}{P_s}
\newcommand{\Pss}{P_{s'}}

\newcommand{\Qt}{Q_t}
\newcommand{\Qs}{Q_s}

\newcommand{\tab}[1]{c\left[#1\right]}
\newcommand{\cost}[1]{\omega\left(#1\right)}
\newcommand{\cosst}{\omega}
\newcommand{\QQt}{\Qt}
\newcommand{\PPt}{\Pt}
\makeatletter
\def\th@plain{%
\thm@notefont{}
\itshape 
}
\def\th@definition{%
\thm@notefont{}
\normalfont 
}
\makeatother

\theoremstyle{definition}
\newtheorem{definition}{Definition}[section]

\theoremstyle{plain}
\newtheorem{theorem}[definition]{Theorem}
\newtheorem{prop}[definition]{Proposition}
\newtheorem{cor}[definition]{Corollary}
\newtheorem{lemma}[definition]{Lemma}

\allowdisplaybreaks

\begin{document}
\title{Homology Localization Through the Looking-Glass of Parameterized Complexity Theory}
\author{Nello Blaser, Erlend Raa V{\aa}gset}
\date{University of Bergen}
\maketitle

\begin{abstract}
Finding a cycle of lowest weight that represents a homology class in a simplicial complex is known as homology localization (HL). Here we address this NP-complete problem using parameterized complexity theory. We show that it is W[1]-hard to approximate the HL problem when it is parameterized by solution size. We have also designed and implemented two algorithms based on treewidth solving the HL problem in FPT-time. Both algorithms are ETH-tight but our results shows that one outperforms the other in practice.


\end{abstract}

\section{Introduction}

Finding and computing topological features in spaces has recently come into prominent attention with the rise of topological data analysis, a novel research area where algebraic topology is applied to data analysis. Topological features are now used in many different fields, such as computational biology \cite{rabadan_blumberg_2019}, neural network analysis \cite{guss2018characterizing} and computer vision \cite{Venkataraman2016}. Topological features are often preferred over purely geometric features because they give qualitative information and reduce the reliance on poorly justified choices of coordinate systems or metrics, leading to more robust results \cite{Carlsson2009}. This makes topological features good candidates for characterizing the shape of data.

Topological features give global characterizations of shape but this characterization is by nature highly algebraic. When the goal is to find local structures, for example in the context of topological noise removal \cite{Guskov2001} or hole detection in sensor networks \cite{Funke2005}, we need to find these global characterizations in our data. One way of doing this is by solving the \textsc{Homology Localization} problem (HL$_d$). This is the problem of finding a $d$-cycle of lowest cost or weight representing a given homology class in a (weighted) simplicial complex. This problem is NP-complete \cite{chen2011hardness} so it has no polynomial time algorithm unless P=NP.

It is possible to deal with NP-completeness without resolving P$\overset{?}{=}$NP through a variety of different methods. Several of these have already been applied to the HL$_d$ problem. One approach is to restrict attention to special cases of the problem, instead of solving it in full generality. This was done by Chen and Freedman in \cite{chen2011hardness}, where they found a polynomial time algorithm for solving the HL$_d$ problem on simplicial complexes embedded in $\R^d$. One can also find approximate solutions like Borradiale et al. did in  \cite{borradaile2020minimum} where they presented various (non-constant factor) approximation algorithm for the HL$_d$ problem on $d+1$ dimensional manifolds and for spaces embedded in $\R^{d+1}$.

These results are close to the limits of what we can hope to achieve within the classical framework. It is NP-hard to find constant factor approximations for the HL$_d$ problem when $d\geq 1$ \cite{chen2011hardness}. The problem remains NP-hard when restricted to $d+1$ dimensional manifolds embedded in $\mathbb{R}^{d+2}$, and it is also hard to approximate in this case under the unique games conjecture \cite{borradaile2020minimum}. This makes it clear that we should approach this problem in a different way.

This paper investigates how and when parameterized algorithms can get us around the many obstacles keeping us from solving the HL problem in practice. This approach is based on a simple yet powerful observation: A problem instance is  hard not because it is big but because it is complicated. The focus of parameterized complexity theory is to figure out which features restrict the complexity of a problem. We do so by studying how fixing the value of parameters, a number measuring properties of a problem instance, impacts how fast a problem can be solved The goal is to find the parameters that puts strong limits on the hardness of a problem and to rule out those parameters that do not. 

We want parameterized algorithms where the exponential explosion in runtime is a function of the parameter alone. Such algorithms can often solve large instances of NP-hard problems fast when the parameter is small. We can also use a complexity hierarchy to show hardness for parameterized problems. These are results that are analogous to NP-hardness and so they are very useful because they tell us which parameters to avoid.

Although this is the first time parameterized complexity theory is systematically applied to the HL$_d$ problem, it is not the first time parameterized algorithms have been used to solve it. In particular, \cite{Busaryev2012} presents a fixed parameter tractable (FPT) algorithm for the HL$_1$ problem using the first homology rank of the simplicial complex as a parameter. This algorithm was later enhanced with various heuristics and used to solve real world instances of the HL$_1$ problem related to cardiac trabeculae reconstruction \cite{wu2017optimal,zhang2019heuristic}. More recently, another FPT-algorithm was put forward to solve the HL$_d$ problem on $d+1$ dimensional manifolds by using the solution size as its parameter  \cite{borradaile2020minimum}.

\subsection{Outline}
\cref{sec:preliminaries} gives a brief introduction to topological notions such as simplicial complexes, homology and the \textsc{Homology Localization} problem, while \cref{sec:ParameterizedComplexity} covers basic concepts from parameterized complexity like FPT algorithms, the W-hierarchy and treewidth. In \cref{sec:SolutionSize} we prove that finding constant factor approximations to the HL$_d$ problem parameterized by solution size is W[1]-hard. In \cref{sec:algorithms} we present two different FPT algorithms for the HL$_d$ problem based on treewidth, prove that they compute the minimal cycles correctly and determine their complexity. In \cref{sec:eth} we prove that the two FPT algorithms are essentially optimal if we assume that the exponential time hypothesis (ETH) is true. In \cref{sec:comparison} we compare the two algorithms based on how they performed when implemented in Python. We end with \cref{sec:conclusion} where we reflect on our results and their implications.

\section{Homology Localization} \label{sec:preliminaries}

This section serves as a brief introduction to key topological concepts, including simplicial complexes, chain complexes, boundary maps, simplicial homology and suspension. This is also where we introduce and formally define the \textsc{Homology Localization} problem. 

\subsection{Simplicial Complexes}
A \emph{(finite) simplicial complex} $K$ is a (finite) family of (finite) sets closed under the subset operation. The elements of $K$ are called \emph{simplicies} and the elements of the simplices are called vertices. We refer to the set $V(K)=\cup_{\sigma\in K}\sigma$ as the \emph{vertex set} of $K$.

A simplex $\sigma$ in $K$ containing precisely $d+1$ vertices from $V(K)$ is called a \emph{$d$-simplex} and we say it has \emph{dimension} $d$. We denote the family of $d$-simplices contained in $K$ by $K_d$. If an $e$-simplex $\rho$ is the subset of the $d$-simplex $\sigma$ then we say that $\rho$ is an \emph{$e$-face} of $\sigma$ and that $\sigma$ is a \emph{$d$-coface} of $\rho$. The \emph{closure} of a family of simplices $X\subset K$, denoted by $\overline{X}$, is the smallest simplicial complex containing all the simplices in $X$. Explicitly, this can be constructed as the set containing all the faces of every simplex in $X$ i.e. $\overline{X}= \cup_{\sigma\in X}({\rho | \rho \subseteq \sigma })$. 
We sometimes visualize a simplicial complex as a triangulated space or shape where $0$-, $1$- and $2$-simplices are represented as points, lines and triangles respectively. \cref{Fig: GeomReal} is an example of this, where we see how the closure of a set can be viewed geometrically as adding all the limit points of a set to itself.

\begin{figure}[h!]
\centering
\includegraphics[width=160mm]{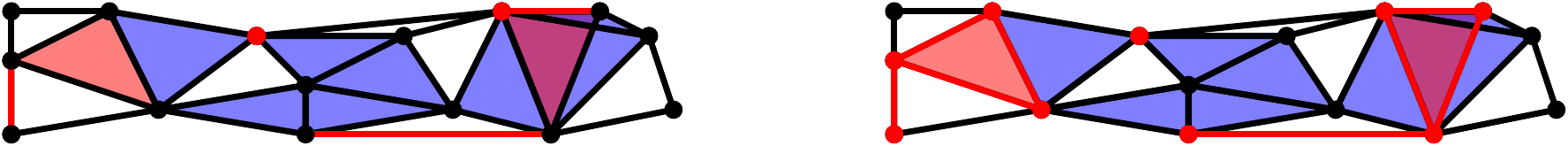}
\caption{Some simplices coloured red in a simplicial complex (left) together with a second copy (right) this time showing the closure of the first set of red simplices in red.
\label{Fig: GeomReal}}
\end{figure}

\subsection{Homology Localization}

Computing simplicial homology is a two step process, the first involving the construction of a \emph{chain complex}. We restrict attention to  chain complexes with coefficents in $\mathbb{Z}_2$ in this paper. These consists of a sequence of vector spaces over $\mathbb{Z}_2$, $C_i$ for $i \geq -1$ where $C_{-1}= 0$. These vector spaces are connected by a sequence of linear transformations called \emph{boundary maps}, $\partial_{i} \colon C_{i} \to C_{i-1}$, for $ i \geq 0$ that have the property that $\partial_{i}\partial_{i-1} = 0$. Given a simplicial complex $K$ we can always construct a chain complex, $C_\bullet(K)$. This is done by letting $C_i(K)$ be the vector space over $\mathbb{Z}_2$ with basis $K_i$. The boundary maps $\partial_i\colon C_{i}(K)\to C_{i-1}(K)$ are then linear extension of the map where $\sigma \in C_i(K)$ is sent to the sum of the $(i-1)$-faces of $\sigma$. That this implies $\partial_{i+1}\partial_i = 0$ is well known.

Vectors in $C_i(K)$ are called \emph{$i$-chains}. Vectors in $Z_i(K)\subset C_i(K)$, the null space  of $\partial_i$, are called \emph{$i$-cycles}. Vectors in, $B_i(K)\subset C_i(K)$, the image of $\partial_{i+1}$, are called \emph{boundaries}. Since $\partial_{i+1} \partial_i = 0$ we have that $B_i(K)\subseteq Z_i(K)$ and so the quotient $H_i(K) = Z_i(K)/B_i(K)$ is well defined. The vector spaces $H_i(K)$ are known as the \emph{homology groups} of $K$. 

\begin{definition} \label{Def: Homologous}
We say that two cycles $u,v\in Z_i(K)$ are \emph{homologous} if they represent the same element in $H_i(K)$, meaning if there is some $(i+1)$-chain $b$ such that $u=v+\partial_{i+1}(b)$. 
\end{definition}

Any $d$-chain $u = \sum a_i \cdot \sigma_i$ in $C_d(K)$ can be identified with the subsets $U = \{\sigma_i \in K_d| a_i = 1 \}$ of $K_d$. More precisely, the map $u\mapsto U$ is an isomorphism of vector spaces over $\mathbb{Z}_2$ if we define addition of $U$ and $V$ as the symmetric difference $U \triangle V$ and scalar multiplication as $U\cdot 1 = U$ and $U\cdot 0 =\emptyset$. 
A simplicial complex is said to be \emph{weighted} if it comes equipped with a function $\cosst \colon K \to \mathbb{R}$ assigning a weight to each simplex. The \emph{cost} of a chain $V \subseteq K_d$ is denoted by \(\cost{V}\) and is defined to be the sum of the weights of the simplices it contains.

\begin{definition} \textsc{Homology Localization} (HL$_d$):\\
Input: A weighted simplicial complex $K$, a $d$-cycle $V$ and a number $k \in \mathbb{R}$.\\
Question: Is there a $d$-cycle $U$ homologous to $V$ such that $\cosst(U) \leq k$?
\end{definition}

We have phrased the HL$_d$ as a decision problem, meaning it outputs either a yes or a no, but the algorithms we have designed returns a cycle $U$ homologous to $V$ of lowest possible weight. \cref{Fig: Homologous cycles} below shows a solution to the optimization version of the HL$_1$. Assuming the weights are all $1$, we can count and see that the input cycle $V$ costs $14$. By adding the boundary of $W$ we get a cycle, $U$, which has size $4$ and there is clearly no other cycle homologous to $V$ of this size or smaller.

\begin{figure}[h!]
\centering
\includegraphics[width=160mm]{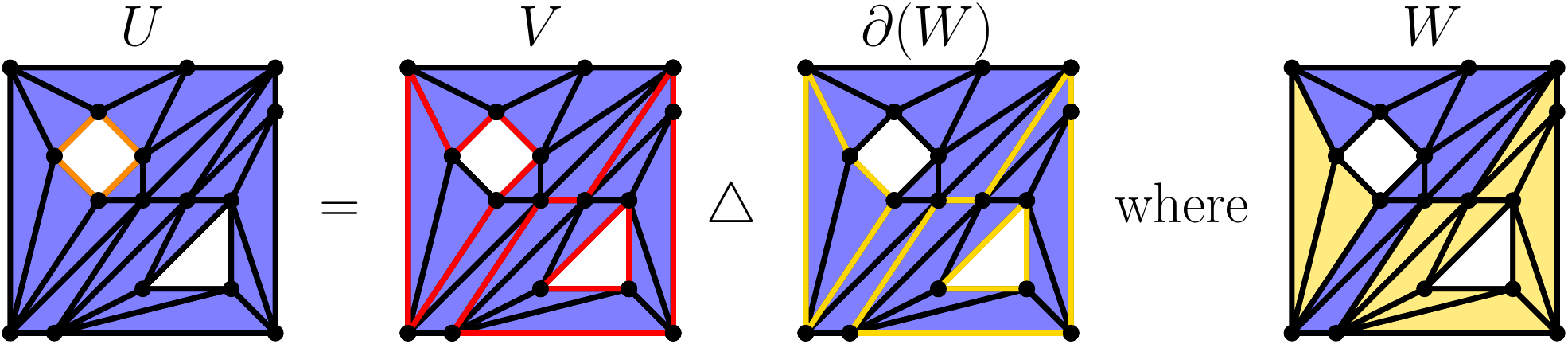}
\caption{\label{Fig: Homologous cycles} An optimal solution, $U$, to an instance, $V$, of the HL$_1$ problem.}
\end{figure}

We will use suspension to reduce the HL$_{d}$ problem to the HL$_{d+1}$ problem in section \cref{sec:SolutionSize} and \cref{sec:eth}.

\begin{definition}
The \emph{suspension} of the simplicial complex $K$ is defined as the simplicial complex $S(K) = \{\sigma, \sigma\sqcup v^+, \sigma\sqcup v_- |\sigma \in K\cup\{\emptyset\} \}\setminus \{\emptyset\}$. 
\end{definition}

There is a one-one correspondence between the $d$-cycles in $K$ and $d+1$-cycles in $S(K)$ given by mapping a $d$-cycle $V$ in $K$ to the $d+1$-cycle $S(V) \in \{\sigma \cup v^+, \sigma\cup v_- | \sigma \in V\}$ in $S(K)$. It is elementary to prove that the two $d$-cycles $U$ and $V$ are homologous in $K$ if and only if $S(U)$ and $S(V)$ are homologous in $S(K)$. These two facts forms the basis for a polynomial time reduction from the (unweighted) HL$_{d}$ to the (unweighted) HL$_{d+1}$ where $(K,k)$ is mapped to $(S(K), 2\cdot k)$. This reduction has several nice properties, one of which is that if $K$ can be embedded in $\mathbb{R}^d$ then $S(K)$ can be embedded in $\mathbb{R}^{d+1}$.

\section{Parameterized Complexity Theory}\label{sec:ParameterizedComplexity}

This section covers the basic ideas from parameterized complexity needed to understand the rest of the paper. In particular we introduce parameterization, FPT-algorithms, W[1]-hardness and treewidth. There are many good books explaining these topics in greater detail, such as \cite{downey2012parameterized, cygan2015parameterized}.

\subsection{Fixed Parameter Tractability}

Classical complexity theory is the study of how a single parameter, the input size, impacts the complexity of a computational task. NP-hardness is perhaps the most important concept from this field. A huge number of both important and useful computational problems are NP-hard. These problems are widely believed to be unsolvable in polynomial time for many different theoretical, empirical and philosophical reasons but this is still (famously) not proven. 

Parameterized complexity theory is the multivariate response to the challenge of overcoming one dimensional NP-hardness in practice. It presents us with a mathematical foundation for analysing the computational complexity of a task with respect to several different parameters at the same time. Using this approach, it is possible to find practical algorithms for many NP-hard problems.

A \emph{parameter} $k$ is a number measuring a property of the problem instance. Virtually anything can be used as the basis for a parameter. Some of the common examples include the size of a solution, how close an approximation is to the answer and how ``structured'' the input is. A problem becomes \emph{parameterized} once we specify which parameter we look at. The value of this parameter is always given as part of the input. 

The runtime of a \emph{parameterized algorithm} is expressed as a function $f(n,k)$, where $n$ is the input size and $k$ is the value of the parameter. Let $p(n)$ denote a polynomial function and $g(k)$ be a computable function. An algorithm is \textit{fixed parameter tractable} (FPT) if it runs in $\mathcal{O}(g(k)\cdot p(n))$-time and it is \textit{slicewise polynomial} (XP) if it runs in $\mathcal{O}(p(n)^{g(k)})$-time. FPT-algorithms are generally preferred over XP-algorithms. It is well documented that FPT-algorithms are practical. To illustrate why XP-algorithms are generally not, we have included \cref{table:xp} from \cite{downey1999parameterized}. In it we have compared the XP runtime $n^kn = n^{k+1}$ with the FPT runtime $2^kn$ using different values of $n$ and $k$.

\begin{table}[h]
\begin{center}
\begin{tabular}{|l||*{5}{c|}}\hline
&\makebox[4em]{$n=50$}&\makebox[4em]{$n=100$}&\makebox[4em]{$n=150$}\\\hline\hline
$k=2$ & 625 & 2,500 & 5,625\\\hline
$k=3$ &15,625&125,000& 421,875\\\hline
$k=5$ &390,625&6,250,000 & 31,640,625\\\hline
$k=10$ &$1.9 \times 10^{12}$&$9.8 \times 10^{14}$&$3.7 \times 10^{16}$\\\hline
$k=20$ &$1.8 \times 10^{26}$&$9.5 \times 10^{31}$&$2.1 \times 10^{35}$\\\hline
\end{tabular}
\caption{\label{table:xp} The ratio $\frac{n^{k+1}}{2^kn}$ for various values of $n$ and $k$.}
\end{center}
\end{table}

More abstractly, if we have one XP-algorithm $h(n,k)$ and one FPT-algorithm $f(n,k)$ then there is a number $k_0$ so that for every $k \geq k_0$ we have $f(n,k) \in \mathcal{O}(h(n,k))$. This shows that the class \textbf{FPT} of problems solvable in FPT-time is contained in the class \textbf{XP} of problems solvable in XP-time.  
NP-hardness can sometimes be used to show that a problem is unsolvable in FPT-time but this does not work when the problem is in \textbf{XP}. To find out if these problems are unsolvable by an FPT-algorithm we have to show that they are W[1]-hard. We think of this as the parameterized analogy of NP-hardness.


\begin{definition}[Parameterized reduction]
Let $\mathcal{A}$ and $\mathcal{B}$ be parameterized problems. A \emph{parameterized reduction} $F\colon \mathcal{A} \to \mathcal{B}$ is a function mapping instances $(X,k)$ of $\mathcal{A}$ to instances $(Y,l)$ of $\mathcal{B}$ in such a way that
\begin{itemize}
    \itemsep0pt
    \item $F(X,k)$ can be computed in FPT-time.
    \item $l \leq g(k)$ for some computable function $g$.
    \item $(X,k)$ is a ``yes'' instance if and only if $(Y,l)$ is a ``yes'' instance.
\end{itemize}
\end{definition}

Assume we have such a parameterized reduction. Then the existence of an FPT-algorithm for $\mathcal{B}$ implies that there is an FPT-algorithm for $\mathcal{A}$. 
Many classes of problems are equivalent in the sense that if one of them is in FPT then so are all of them. One of these classes is the class of \emph{W[1]-complete} problems. A problem is said to be W[1]-hard if there is a parameterized reduction to it from a W[1]-complete problem. There is no known FPT-algorithm for any W[1]-hard problem and most researchers in the field believe that $\textbf{FPT} \neq W[1]$. This is currently unproven but there are many theoretical, empirical and philosophical reasons for believing that these sets are not the same.   


\subsection{Tree Decompositions}

A \emph{graph} $G$ is a 1-dimensional simplicial complex. The 1-simplices are called edges and the edge containing the vertices $u$ and $v$ is denoted by $uv$. A \emph{tree} $T$ is a graph where $H_0(T)= \mathbb{Z}_2$ and $H_1(T)= 0$. The vertices of trees are often called \emph{nodes}. A \emph{path} from $u$ to $v$ in $G$ is a sequence $(u = u_1,\dots, u_k= v)$ where $v_iv_{i+1} \in G$ for $1\leq i \leq k$ and where every vertex is unique. Trees can also be defined as graphs where every pair of nodes $t,t' \in T$ are connected by a unique path, $\mathtt{path}(t,t')$. A \emph{rooted tree} $(T,r)$ is a tree $T$ together with a node $r$ which we is given a special status for bookkeeping purposes. Every node, $t$, in a rooted tree partitions the other nodes into two sets $\{s \in T | t\in \mathtt{path}(s,r)\}$ and $\{s \in T | t\not \in \mathtt{path}(s,r)\}$. The nodes in the first set are called the \emph{descendants} of $t$. The first node on the path from $v$ to the root is the \emph{parent} of $v$ and $v$ is the child of that node. Every node has precisely one parent except for the root which has none. Nodes with no children are called \emph{leaves}.

\begin{definition}[Nice Tree Decomposition]
A \emph{tree decomposition} of $G$ is a function $X_{-} \colon V(T) \to \mathcal{P}(V(G))$ where $T$ is a tree. The nodes $t\in T$ are mapped to sets $X_t$ of vertices in $G$ called \emph{bags} such that

\begin{itemize}
    \itemsep0em 
    \item $\forall v\in G$ there exists $t\in T$ such that $v\in X_t$.
    \item $\forall uv\in G$ there exists $t\in T$ such that $u, v\in X_t$.
    \item if $u\in X_t,X_{t'}$ for $t,t' \in T$ then $u\in X_{s}$ for every vertex $s\in \mathtt{path}(t,t')$.
\end{itemize}
A tree decomposition is \emph{nice} if $X_r = \emptyset$ and every bag $X_t$ is either 
\begin{itemize}
    \itemsep0em 
    \item a \emph{leaf bag} where $t$ is a leaf and $X_t = \emptyset$.
    \item an \emph{introduce bag} where $t$ has a child, $s$, and $X_t = X_s\sqcup \{v\}$.
    \item a \emph{forget bag} where $t$ has a child, $s$, and $X_t\sqcup \{v\} = X_s$.
    \item or a \emph{join bag} where $t$ has two children, $s$ and $s'$, and $X_t = X_s = X_{s'}$.
\end{itemize}
\end{definition} 

The \emph{width} of a tree decomposition is defined as $\max_{t\in T}(||X_t||)-1$. The \emph{treewidth} of a graph is the width of the tree decomposition of that graph of smallest width. The empty graph has a treewidth of $-1$, a non-empty graph without edges has treewidth $0$, trees have treewidth $1$ and the complete graph has treewidth $||V||-1$. The tree decomposition of the graph pictured in \cref{FIG: Tree decomposition} below has width $2$ and this is also the treewidth of that graph. Every tree decomposition can be transformed into a nice tree decomposition by increasing the number of bags by a constant factor and without increasing the width. A nice tree decomposition of the graph in \cref{FIG: Tree decomposition} is shown later in \cref{FIG: Tree decomposition nice ETH} in \cref{sec:eth}.

\begin{figure}[!h]
\centering
\begin{subfigure}{.28\textwidth}
  \centering
  \includegraphics[width=0.95\linewidth]{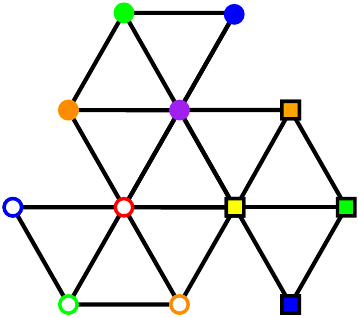}
\end{subfigure}%
\begin{subfigure}{.28\textwidth}
  \centering
  \includegraphics[width=0.85\linewidth]{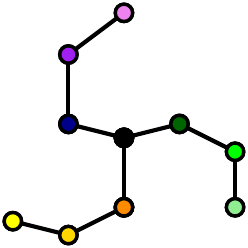}
\end{subfigure}
\begin{subfigure}{.28\textwidth}
  \centering
  \includegraphics[width=0.9\linewidth]{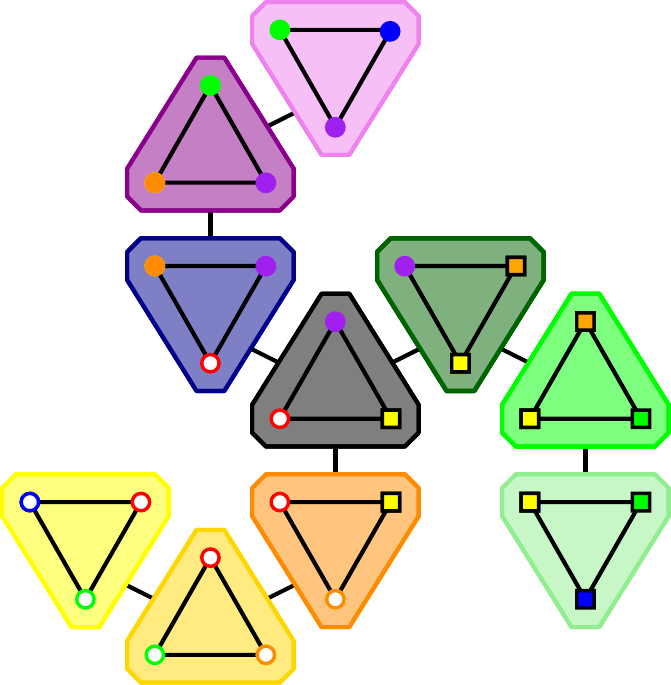}
\end{subfigure}
 \caption{\label{FIG: Tree decomposition}A graph (left) together with a graphical representation of a tree decomposition (right). The figure in the middle is the tree used in the tree decomposition and its nodes are mapped to the content of the ``triangle'' that has approximately the same colour and relative position.}
\end{figure}

Computing the treewidth of a graph is an NP-complete problem but this is not going to be a big issue. The problem can be solved in FPT-time when the parameter $k$ is the solution size (i.e. the treewidth of the input graph) \cite{bodlaender1996linear}. There is also an algorithm running in $2^{\mathcal{O}(k)}\cdot n$-time (where $k$ is the treewidth of the graph) that finds a tree decomposition whose width is a constant factor approximation of the actual treewidth \cite{bodlaender2016c}. This means that any problem solvable in FPT-/$2^{\mathcal{O}(k)}\cdot n$-time when parameterized by the width of a tree decomposition that is given as part of the input is also solvable in FPT-/$2^{\mathcal{O}(k')}\cdot n$-time when parameterized by the treewidth, $k'$, of the actual graph. On the practical side of things, the Parameterized Algorithms and Computational Experiments (PACE) Challenge in 2017 focused on finding the treewidth of graphs which resulted in numerous practical implementations \cite{PACE}. 

We stick to the common practice of ignoring the additional cost of finding a tree decomposition for the above reasons. This is achieved by just assuming that some tree decomposition of the graph of width $k$ is given together with the input.

\section{Solution Size as a Parameter}\label{sec:SolutionSize}

This is the section where we prove that the \textsc{HL}$_d$ problem is W[1]-hard when parameterized by solution size for all $d\geq 1$. In fact, we prove the even stronger result that the problem is W[1]-hard to approximate to a constant factor using this parameter. These results even hold when the input is restricted to simplicial complexes embedded in $\mathbb{R}^{d+3}$. The problems in this section are all parameterized by solutions size.

\subsection{Parameterized Gap Problems}

An elegant way of proving that a problem can not be approximated in FPT-time is by showing that the gap version of the problem is W[1]-hard. 

\begin{definition} \textsc{Gap Nearest Codeword} (NC$_\gamma$)\\
Input: An $m\times n$ matrix $A$ over $\mathbb{Z}_2$, a vector $v\in \mathbb{Z}_2^n$ and a number $k\in \mathbb{R}$.\newline
Output:``YES'' if there is a vector $b\in \mathbb{Z}_2^n$ such that $|| v + Ab|| \leq k$ and ``NO'' if for every vector $b\in \mathbb{Z}_2^n$ we have $||v + Ab|| > k\cdot \gamma$. Otherwise the output does not matter.
\end{definition}

The \textsc{Nearest Codeword} (NC) problem is defined as NC$_1$ and this is just a normal decision problem. When $\gamma > 1$ this is no longer the case, since we do not need to answer ``yes'' or ``no'' on every input. Formally, gap problems are examples of \emph{promise problems}. As the name suggests, these are problems where we are ``promised'' that an instance has certain (good) properties. In the particular case of gap problems, the property is that an instance is either a yes instance or that it is not even close to being one in terms of solution size. Promise problems are more general than decision problems but the concept of $W[1]$-hardness generalizes to this setting without complications.

\begin{theorem}[\cite{bhattacharyya2019parameterized}]\label{THM: NC hardness of natural parameter}
The NC$_{\gamma}$ problem parameterized by solution size is $W[1]$-hard.
\end{theorem}

The HL$_d$ problem has a gap version as well. 

\begin{definition} \textsc{Gap Homology Localization} (HL$_{d,\gamma}$)\\
Input: A weighted simplicial complex $K$, a $d$-cycle $V$ and a number $k \in \mathbb{R}$. \newline
Output: ``YES'' if there exists a cycle $U$ homologous to $V$ weighing less than $k$ and ``NO'' if every cycle $U$ homologous to $V$ weighs more than $k \cdot \gamma$. Otherwise the output does not matter.
\end{definition} 

The \emph{unweighted} HL$_{d,\gamma}$ problem is the restricted version of the HL$_{d,\gamma}$ problem where the weight of every simplex in the simplicial complex is $1$. 
 
\begin{definition} An FPT reduction from a parameterized promise problem, $\mathcal{A}$, to another, $\mathcal{B}$, is a procedure that transforms $(X, k)\in \mathcal{A}$ to $(Y, k') \in \mathcal{B}$ in FPT-time so that:
\begin{itemize}
    \itemsep0pt 
    \item There exists a computable function $g$ such that $k' \geq  g(k)$.
    \item If the input is a ``YES'' instance then so is the output instance.
    \item If the input is a ``NO'' instance then so is the output instance.
\end{itemize}
\end{definition}

If there is a reduction of this kind from a $W[1]$-hard promise problem to another promise problem then this other promise problem is $W[1]$-hard as well. Note that a gap problem can be solved with a single call to a $\gamma$-approximation algorithm: First compute the size of an approximate solution. If this is less or equal to $k\cdot \gamma$ output ``YES'', otherwise output ``NO''.
\subsection{W[1]-Hardness of Approximation}

\begin{theorem}\label{THM: hardness of natural parameter}
The unweighted HL$_{1, \gamma}$ problem parameterized by solution size is $W[1]$-hard.
\end{theorem}
\begin{proof} All we need is to make some observations about the \emph{strict} reduction used in \cite{chen2011hardness}. We provide a brief sketch of the parts that we need, in particular the polynomial reduction from the $NC$ problem to the HL$_1$ problem. First the matrix $A$ is used to define a space $K$. Each row of the matrix is represented as a topological circle while each column is represented by a $2$-sphere with one hole for every row in which the column takes the value $1$. The boundaries of these holes are glued to the circle representing that particular row. $K$ can then be triangulated so that it becomes a simplicial complex with some important properties. For instance, each ``row circle'' is a subcomplex containing precisely $4$ edges. The vector $v$ is mapped to the cycle $V$ which is made up of the set of edges in the circles that corresponds to the rows where $v$ is $1$.

\begin{figure}[!h]
\centering
\begin{subfigure}{.30\textwidth}
  \centering
  \includegraphics[width=0.9\linewidth]{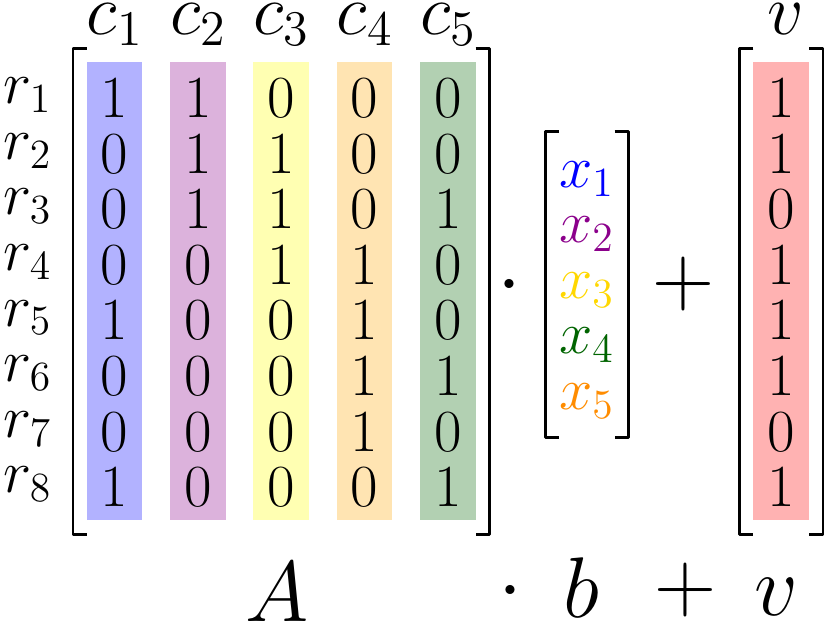}
\end{subfigure}%
\begin{subfigure}{.65\textwidth}
  \centering
  \includegraphics[width=0.9\linewidth]{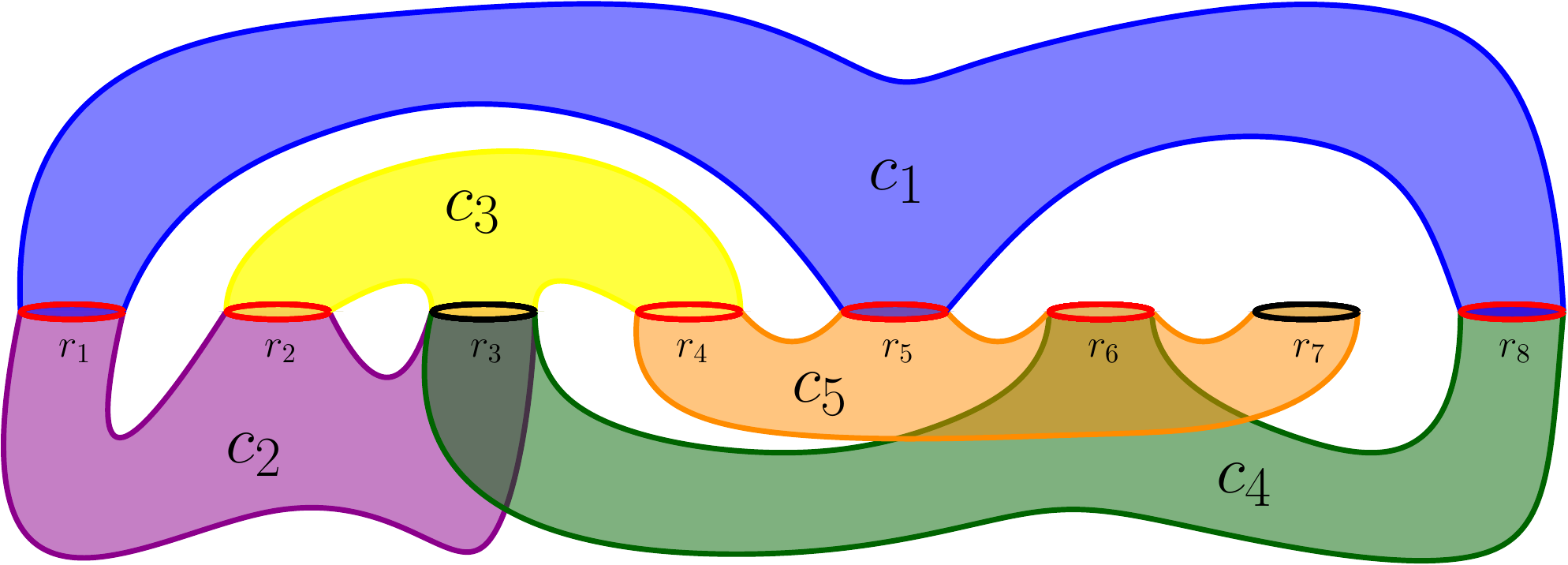}
\end{subfigure}
\caption{\label{FIG: Chen reduction} The FPT-reduction associate a circle to each row $r_i$ and each column $c_i$ is identified with a $2$-sphere with $||c_i||$ open discs removed. Each boundary component of a ``column sphere'' is glued to a unique row circle corresponding to a row where $c_i$ is $1$. The vector $v$ is mapped to an input cycle $V$ marked in red and it consists of the $1$-simplices in the circles corresponding to rows where $v$ is $1$.}
\end{figure}

The optimal solution to an input instance is always precisely one fourth of the size of an optimal solution to the output instance. First, note that any solution to the input problem gives a solution to the output that is $4$ times bigger. Just let $U = V \triangle \delta(W)$ where $W$ consists of every $2$-simplex inside the ``$2$-spheres'' that corresponds to the columns used in the optimal solution to the input problem. To show the converse we will use the fact that any cycle $U$ homologous to $V$ is itself homologous to a third cycle $U'$ which is no larger than $U$ but is contained within the ``row circles''. This was proven constructively in \cite{chen2011hardness}. If we know that a cycle $U'$ has these these properties then $V$ is homologous to $U'$ via a $2$-chain $W$ containing either none or all of the $2$-simplices in any ``column-sphere''. This way we also get a solution to the input problem that is $1/4$ of the size of $U'$. 
If we map $k$ to $k' = g(k) = 4\cdot k$ then this yields a parameterized reduction from the NC$_\gamma$ problem to the HL$_{1,\gamma}$ problem. 

\end{proof}

\begin{cor}
The unweighted HL$_{d, \gamma}$ problem parameterized by solution size is $W[1]$-hard for $d\geq 1$ even when it is restricted to spaces embedded in $\mathbb{R}^{d+3}$.
\end{cor}
\begin{proof}

The proof is by induction. For the base case, note that the space in the proof of \cref{THM: hardness of natural parameter} can always be embedded in $\mathbb{R}^4$. We can use a generalization of book embeddings of graphs to see this. Let the spine of such a book be $2$-dimensional and let the pages be $3$-dimensional. First we embed the ``row circles'' in the spine. Then embed each ``column $2$-sphere'' on its own page so that only its boundary intersects the spine in such a way that each boundary component is identified with the appropriate ``row circle''. This shows that the space can be embedded in the book which in turn can be embedded in $\mathbb{R}^4$.

The induction step uses the reduction from the HL$_{d}$ problem to the HL$_{d+1}$ problem from \cref{sec:preliminaries} based on suspension. This is clearly a parameter preserving reduction and the bijection between feasible solutions can be used to show that this reduction is also a gap preserving reduction from the HL$_{d, \gamma}$ problem to the HL$_{d+1, \gamma}$ problem.
\end{proof}

\section{Treewidth Algorithms}\label{sec:algorithms}



The two algorithms we present here work on the same principle. We take the width of a tree decomposition of a graph related to the input complex as a parameter. We use the decomposition to traverse the simplicial complex in an orderly fashion computing a large number of ``partial solutions'' to the input problem. These partial solutions cover larger and larger portions of the input. Each new partial solution is computed dynamically based on those solutions we have already found by either extending them or gluing them together. In the end the entire complex has been processed. The final partial problem we solve will be equivalent to solving the HL$_d$ problem. These partial problems we need to solve to get to this point are all instances of the \textsc{Restricted HL$_d$} (R-HL$_d$) problem.

\begin{definition}\label{Def: Restricted HL} \textsc{Restricted HL$_d$ (R-HL$_d$)}: \newline 
Input: A simplicial complex $K$, a $d$-cycle $V$ in $K$ and a four-tuple $(G_t, X_t,\Qt, \Pt )$ where $G_t\subseteq K_d \cup K_{d+1}$, $X_t\subseteq G_t$, $\Qt \subseteq X_t \cap K_{d+1}$ and $\Pt \subseteq X_t \cap K_d$ \newline 
Output: The minimum value of $\cost{U\cap F_t}$ where $F_t = G_t \setminus X_t$ and
\begin{itemize}
    \itemsep0pt
    \item $U$ is homologous to $V$ through a $d+1$-chain, $W$, contained in $G_t$.
    \item $U$ intersects $X_t$ at $\Pt $.
    \item $W$ intersects $X_t$ at $\Qt $.
\end{itemize}
If no such $U$ exists then output $\infty$.
\end{definition}

\begin{figure}[h!]
\centering
\includegraphics[width=120mm]{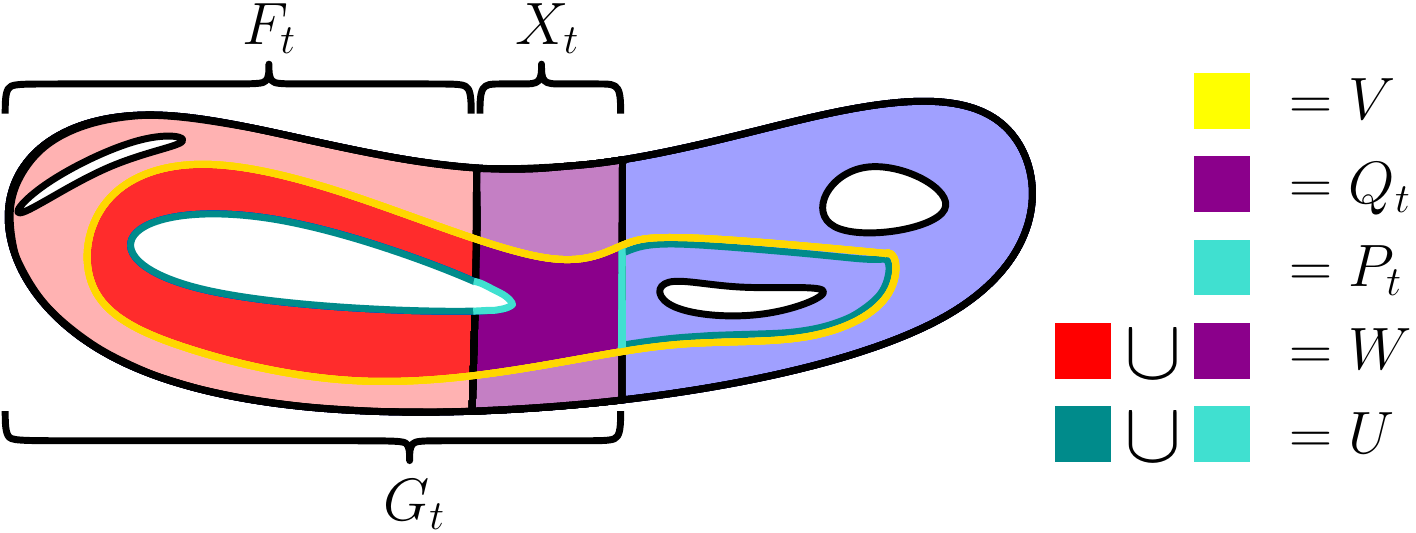}
\caption{A sketch of an instance of the R-HL$_1$ in a a simplicial complex $K$ on input $(G_t,X_t.Q_t,P_t)$ together with a solution $(W,U)$. This is an optimal solution to the instance assuming the weight of a cycle in the space is the same as its length.}
\end{figure}

The idea is that we solve the problem on $G_t$ and ignore the rest of the space. On the subset $X_t$ of $G_t$ we have complete control of of how a solution will behave. This is because the conditions above force $W$ to be precisely $\Qt$ on $X_t$ and $U$ to be precisely $\Pt$ on $X_t$. In this sense, the problem is to find the best solution in $F_t$ that can be spliced together with $\Qt$ and $\Pt$.

\subsection{The Connectivity Graph}\label{sec:alg1}

The first of our two algorithms is parameterized by the treewidth of the $d$-connectivity graph of the simplicial complex.

\begin{definition}
The \emph{ $d$-connectivity graph} of a simplicial complex is a graph with $d$-simplices as vertices and edges going between distinct $d$-simplices that share a common $(d-1)$-face.
\end{definition}

\begin{figure}[h!]
\centering
\includegraphics[width=100mm]{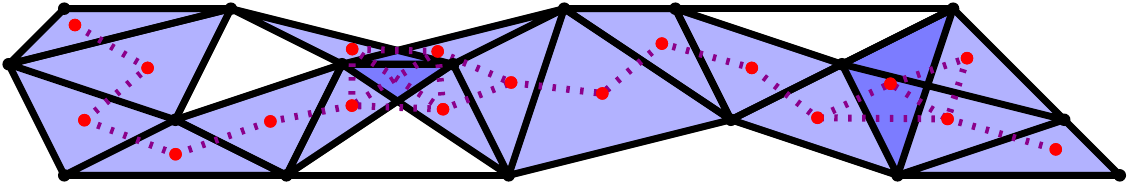}
\caption{\label{FIG: ALG aid}The $2$-connectivity graph in a $2$-dimensional simplicial complex.}
\end{figure}

We denote the $d$-connectivity graph of $K$ as $Con_{d}(K)$. The remainder of this section is used to prove the following theorem.

\begin{theorem}
The HL$_d$ problem can be solved in $\mathcal{O}(2^{k(2d+5)}\cdot n)$ time if we are given a tree decomposition of $Con_{d+1}(K)$ with width $k$.
\end{theorem}
Since a tree decomposition of optimal width can be found in FPT-time \cite{bodlaender1996linear} we also have the following result.
\begin{cor}
The HL$_d$ problem parameterized by the treewidth of $Con_{d+1}(K)$ is in \textbf{FPT}.
\end{cor}

\subsubsection{Connectivity Graph Algorithm}

We will find an optimal solution, $\tab{t,\Qt, \Pt}$, to the R-HL$_d$ for every four-tuple $(G_t,X_t \cup (\overline{X_t})_d,\Qt ,\Pt )$ at every bag $X_t$ in the nice tree decomposition of $Con_{d+1}(K)$. Here $G_t$ denotes the set containing all $d$-/$d+1$-simplices that is in some $\overline{X_s}$, where $s$ is either $t$ or a descendant of $t$. Meanwhile $\Qt$ is any subset of $X_t$ and $\Pt$ is any subset of $(\overline{X_t})_d$ .

The computation is done dynamically, which just means that we will use the solutions for every possible R-HL$_d$ problem at the child bag(s) of $X_t$. This means that the bags of the tree decomposition has to be processed so that we have already computed these solutions. It is elementary to check that processing the bags by the post-ordering obtained from performing a depth first search (DFS) starting at the root will guarantee this. 

We are now ready to describe how to compute a particular solution at the different types of bags. It may be helpful to have a look at both \cref{FIG: ALG aid} and \cref{FIG: ALG1 aid} to get some intuition for what we are doing.
\newline

\textbf{Leaf Node}: 
\newline
Set $\tab{t,\emptyset,\emptyset}$ to $0$.
\newline

\textbf{Introduce node}:
\newline
Suppose $t$ is an introduce node with the child node $s$ and introducing the $d+1$-simplex is $\sigma$, so that $X_t = X_{s} \cup \{\sigma\}$. We let $N_t$ be the set of $d$-simplices ``new'' to the bag, i.e. those $d$-simplices that are in $\overline{X_{t}}$ but not in $\overline{X_{s}}$.  If $\sigma\not \in \Qt $ then we compute

$$
\tab{t,\Qt ,\Pt } = 
\begin{cases}
\tab{s,\Qt ,\Pt \cap \overline{X_{s}}} & \text{if } \Pt \cap N_t= V \cap N_t\\
\infty &  \text{if }\Pt \cap N_t\neq V \cap N_t
\end{cases} 
$$
otherwise, if $\sigma \in \Qt $, we have
$$
\tab{t,\Qt ,\Pt } = 
\begin{cases}
\tab{s,\Qt \cap X_{s} , (\Pt \triangle \partial(\sigma))\cap\overline{X_{s}}} &  \text{if } \Pt \cap N_t = (V \triangle \partial(\sigma))\cap N_t\\
\infty &  \text{if } \Pt \cap N_t \neq (V \triangle \partial(\sigma))\cap N_t.
\end{cases} 
$$

\textbf{Forget node} :
\newline
Suppose $t$ is a forget node with child $s$ such that $X_t= X_{s}\setminus\{\sigma\}$. Let $O_t$ denote the set of ``old'' $d$-simplices, i.e. those $d$-simplices in $\overline{X_{s}}$ that are not in $\overline{X_t}$. Then 
$$\tab{t,\Qt ,\Pt } = \min_{\Qs  \in \mathcal{W},\Ps  \in \mathcal{U}}\big\{\tab{s,\Qs ,\Ps } + \cosst(O_t\cap \Ps) \big\},$$ 
where $\mathcal{W}$ is the family of subset of $X_{s}$ whose restriction to $X_t$ is $\Qt $ and $\mathcal{U}$ is the family of subset of $(\overline{X_{s}})_d$ whose restriction to $\overline{X_t}$ is $\Pt $.
\newline

\textbf{Join node}:
\newline
Suppose that $t$ is a join node and that it has the two child nodes, $s$ and $s'$, where $X_t=X_{s}=X_{s'}$. Let $\Qt $ be a subset of $X_t$ and $\Pt $ be a subset of the $d$-simplices of $\overline{X_t}$ as before. Then
$$\tab{t,\Qt ,\Pt } = \min_{\Ps, \Pss \subseteq {(\overline{X_t})}_d}\big\{\tab{s,\Qt,\Ps}+\tab{s',\Qt,\Pss} \Big | \Pt  = \Ps\triangle \Pss\triangle \partial(\Qt)\triangle (V\cap \overline{X_t})\big\}.$$
\newline

\begin{figure}[h!]
\centering
\includegraphics[width=140mm]{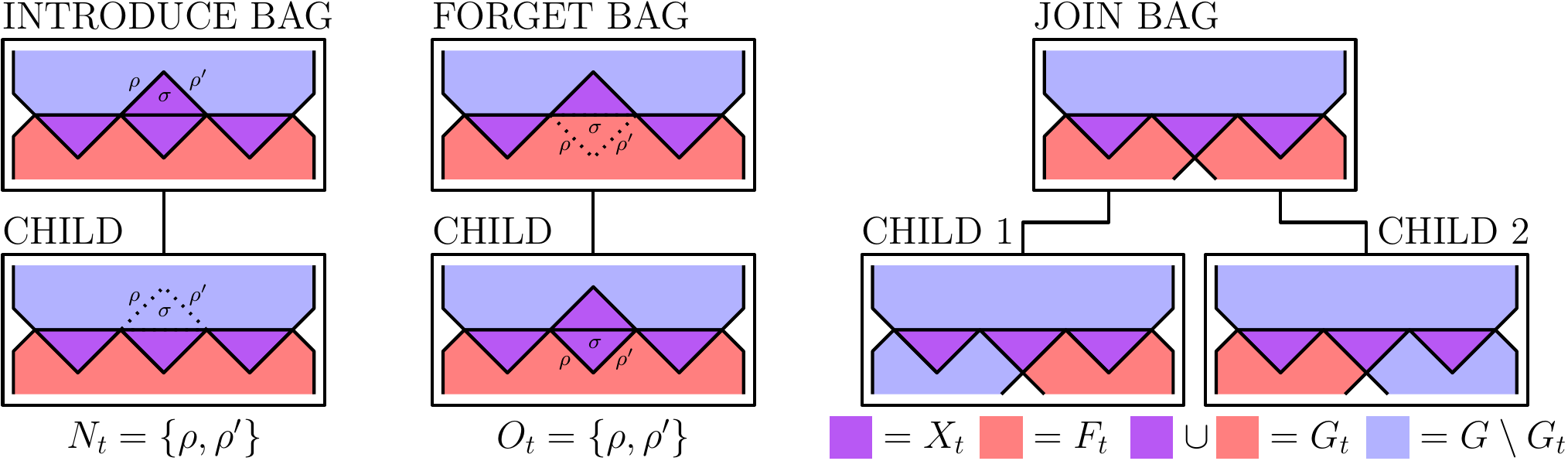}
\caption{\label{FIG: ALG1 aid}Examples of the various sets in a simplicial complex at different types of bags.}
\end{figure} 

\textbf{Unprocessed $d$-Simplices}:
\newline
The $d$-simplices of the input cycle that does not have a $d+1$-coface are not ``seen'' by the algorithm. To get the right cost we need to add the weight of these simplices to the solution we get at the root.

\subsubsection{Correctness}

We carry over the same setting and notation we used in the algorithm. We also define $\mathcal{S}[t,\Qt ,\Pt ]$ as the set of solutions to the R-HL$_d$ with input tuple $(G_t,X_t,\Qt ,\Pt )$. For convenience, we will let elements of $\mathcal{S}[t,\Qt ,\Pt ]$ be pairs consisting of a $d+1$-chain and a $d$-cycle, $(W,U)$. Note that this is somewhat redundant as $W$ determines $U$ uniquely. 
\newline

\textbf{Leaf Node}:
\newline
If $t$ is a leaf node then $(\emptyset, V) \in \mathcal{S}[t,\emptyset ,\emptyset ]$ is a valid solution and since $F_t = \emptyset$ we have that $\cost{U\cap F_t} = 0$ for any cycle $U$.
\newline

\textbf{Introduce Node}:
\newline
First note that $\tab{t,\Qt ,\Pt } = \infty$ is by definition the same as  $S[t,\Qt ,\Pt ] = \emptyset$. We claim that 
\begin{enumerate}
    \itemsep0pt
    \item if $\sigma\not \in \Qt $ and $\Pt \cap N_t \neq V \cap N_t$ then $S[t,\Qt ,\Pt ] = \emptyset$
    \item if $\sigma \in \Qt $ and $\Pt \cap N_t \neq (V \triangle \partial(\sigma))\cap N_t$ then $S[t,\Qt ,\Pt ]=  \emptyset$.
\end{enumerate}
We leave out the details of the proof, which hinges on the observation that $\sigma$ is the only simplex in $G_t$ that is a coface of any of the simplices in $N_t$. This observation follows from the basic properties of a tree decomposition.

There are two things left to prove. Firstly that if $\sigma \not \in \Qt $ and $\Pt \cap N_t=V \cap N_t$ then $\tab{t,\Qt ,\Pt }= \tab{s,\Qt ,\Pt \cap \overline{X_{s}}}$ and secondly that when $\sigma \in \Qt $ and $\Pt \cap N_t = (V \triangle \partial(\sigma))\cap N_t$, then $\tab{t,\Qt ,\Pt } = \tab{s,\Qt \cap X_{s} , (\Pt \triangle \partial(\sigma))\cap\overline{X_{s}}}$. It is easy to see that we can prove the first and second claim respectively by showing that $(W,U) \in \mathcal{S}[t,\Qt  ,\Pt ]$ if and only if $(W,U) \in \mathcal{S}[s,\Qt  ,\Pt \cap \overline{X_{s}}]$  and that $(W,U) \in \mathcal{S}[t,\Qt  ,\Pt ]$ if and only if $(W \setminus \{\sigma\},U\triangle \partial(\sigma)) \in \mathcal{S}[s,\Qt \cap X_{s} , (\Pt \triangle \partial(\sigma))\cap\overline{X_{s}}]$. That these claims are true follows from elementary set theory. 
\newline 

\textbf{Forget Node}:
\newline
The formula follows if we can show that $(W,U) \in \mathcal{S}[t,\Qt  ,\Pt ]$ if and only if there are sets $ A \subseteq \{\sigma\}$ and $B \subseteq O_t$ such that $(W,U) \in \mathcal{S}[s, \Qt \cup A  ,\Pt  \cup B] $. The backwards direction follows from the observation that the latter problems have more restrictions on it than the first and so a solution to any of the latter problems gives a solution to the first. Conversely, if $(W,U) \in \mathcal{S}[t,\Qt  ,\Pt ]$ then we see that $(W,U) \in \mathcal{S}[t,\Qt \cup (W \cap \{\sigma\} ,\Pt  \cup (U \cap O_t])$. We need to add $\cosst(O_t\cap \Ps)$ to the cost because we have to account for the weight of the simplices in the part of the ``old'' solutions intersecting $O_t$. These simplices were not in $F_s$ but they will be in $F_t$.
\newline

\textbf{Join Node}:
\newline
Similarly to before we want to prove that $(W,U) \in \mathcal{S}[t,\Qt  ,\Pt ]$ if and only if there exists a pair of solutions $(W', U') \in \mathcal{S}[s, \Qt  ,\Ps ] $ and $(W'',U'') \in \mathcal{S}[s', \Qt  ,\Pss]$ such that $ \Pt  = \Ps \triangle \Pss \triangle (V \cap \overline{X_t}) \triangle \partial(\Qt )$. This might not be very intuitive so we provide some further details on how to prove it. 

For $(W,U)\in \mathcal{S}[t,\Qt  ,\Pt ]$ we have that $(W\cap \overline G_{s}, (V\cap \overline{G_{s}})\triangle (\partial(W\cap G_{s})) \in  \mathcal{S}[s, \Qt  ,\Ps ]$, where 
$ \Ps = \big(V\cap \overline{X_t}\big) \triangle \big( \partial(W\cap F_{s}) \cap \overline{X_t}\big) \triangle \big(\partial(W\cap X_{t})\big).$ An analogous computation can be made for the other child node. We can then prove that
\begin{align*}
    \Ps \triangle \Pss  &= \left( \partial(W\cap F_{s}) \cap \overline{X_t}\big) \triangle \big( \partial(W\cap F_{s'}) \cap \overline{X_t}\right)\\
                        &= \Pt  \triangle (V \cap \overline{X_t}) \triangle \partial(\Qt )
\end{align*}
with elementary set theory. This implies that the relation $ \Pt  = \Ps \triangle \Pss \triangle (V \cap \overline{X_t}) \triangle \partial(\Qt ) $ holds.

Conversely, if $(W',U')\in \mathcal{S}[s,\Qt  ,\Ps ]$ and $(W'',U'')\in \mathcal{S}[s',\Qt  ,\Pss]$ then we can show that $(W'\cup W'', \partial(W'\cup W'') \triangle (V \cap \overline{G_t}) \in  \mathcal{S}[t, \Qt  ,\Pt ]) $. To see that  $ \Pt  = \Ps \triangle \Pss \triangle (V \cap \overline{X_t}) \triangle \partial(\Qt )$ we just need to prove that 
\begin{align*}
\Pt  &= \left(\partial(W'\cup W'') \triangle (V \cap \overline{G_t})\right) \cap \overline{X_t}\\
    &= \Ps \triangle \Pss \triangle (V \cap \overline{X_t}) \triangle \partial(\Qt ).
\end{align*}

\subsubsection{Runtime Analysis}
The leaf nodes takes constant time to process. In order to process an introduce bag we have to fill at most $2^{k}\cdot 2^{k(d+2)}= 2^{k(d+3)}$ table entries. This is because each of the $d+1$-simplices in the bag can be in one of two states (it can be in $\Qt$ or not). Each such simplex might in turn be a coface of at most $d+2$ $d$-simplices. Each of those can also be in one of two states (either it is in $\Pt$ or it is not). To fill one entry takes constant time assuming it takes constant time to look up a previous solution. Processing the introduce bag therefore takes $\mathcal{O}(2^{k(d+3)})$ time. By the same reasoning we need to fill up to $2^{k}\cdot 2^{k(d+2)}= 2^{k(d+3)}$ table entries in order to process a forget bag. In order to fill an entry we need to compute the minimum of at most $2\cdot(d+2)$ entries from the child bag but this is only a constant number. Thus processing the forget bag also takes $\mathcal{O}(2^{k(d+3)})$ time.

To process the join bag we also need to fill in $2^{k}\cdot 2^{k(d+2)}= 2^{k(d+3)}$ table entries. There are at most $2^{(d+2)k}$ pairs of sets, $\Ps$ and $\Pss $, such that $\Pt  = \Ps\triangle \Pss \triangle \partial(\Qt)\triangle (V\cap \overline{X_t})$. This means that in order to compute the cost of an entry we need to take the minimum of at most $2^{k(d+2)}$ numbers. Each of these number can be computed in constant time and so we need $\mathcal{O}(2^{k(d+3)}\cdot 2^{k(d+2)}) = \mathcal{O}(2^{k(2d+5)})$ time to process the join bag. As there are $\mathcal{O}(n)$ bags to process this results in a total runtime of $\mathcal{O}(2^{k(2d+5)}n)$ in the worst case scenario.

\subsection{The Hasse Diagram}\label{sec:alg2}

The second algorithm is very similar to the first with the difference that we are now parameterizing by the treewidth of a subgraph of the Hasse diagram of the simplicial complex.

\begin{definition}
The \emph{Hasse diagram} of a simplicial complex is the graph with simplices as vertices and edges going between every $i$-simplex and each of its $(i-1)$-faces. By \emph{level $d$ of the Hasse diagram} we mean the full subgraph of the Hasse diagram induced on the vertices corresponding to the $d$- and $(d-1)$-simplices.
\end{definition}

\begin{figure}[h!]
\centering
\includegraphics[width=100mm]{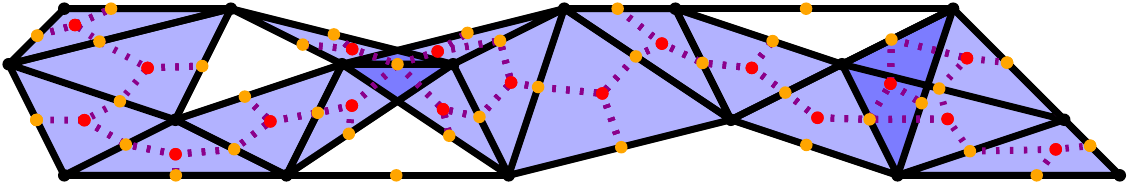}
\caption{A $2$-dimensional simplicial complex with level $2$ of the Hasse diagram embedded.}
\end{figure}

We denote level $d$ of the Hasse diagram of a simplicial complex by $Hasse_{d}(K)$.

\begin{theorem}
The HL$_d$ problem can be solved in $\mathcal{O}(2^{2k}\cdot n)$ time if we are given a tree decomposition of $Hasse_{d}(K)$ with width $k$.
\end{theorem}

\begin{cor}
The HL$_d$ problem parameterized by the treewidth of $Hasse_{d}(K)$ is in \textbf{FPT}.
\end{cor}

\subsubsection{Hasse Diagram Algorithm}
Assume we are given a nice tree decomposition of $Hasse_{d+1}$. At every bag, $X_t$, we will solve the R-HL$_d$ for every input tuple $(G_t,X_t,\Qt ,\Pt )$.  $G_t$ will be the union of all sets $X_s$ such that $s$ is a descendant of $t$, $X_t$ is the bag at node $t$ of the tree decomposition, $\Qt$ is a subset of the $d+1$-simplices of $X_t$ and $\Pt$ is a subset of the $d$-simplices of $X_t$.

We will denote the size of an optimal solution to this problem as $\tab{t,\Qt, \Pt}$ and its value is computed in a similarly to how we did it in \cref{sec:alg1}. The main difference is that we forget and introduce $d$- and $d+1$-simplices one by one, see \cref{FIG: ALG2 aid}. 
\newline

\textbf{Leaf node}:
\newline
Set $\tab{t,\emptyset,\emptyset}$ to $0$.
\newline

\textbf{Introduce node}:\newline
Let $t$ be an introduce node with child $s$. There are then two cases. Either the introduced vertex corresponds to a $d+1$-simplex $\sigma$ so that the $X_t = X_{s} \cup \{\sigma\}$. We then have
\begin{equation*}
    \tab{t,\QQt,\PPt } =
     \begin{cases}
               	\tab{s,\QQt,\PPt }          	& \sigma \not\in \QQt\\
               	\tab{s,\QQt \setminus\{\sigma\} ,(\PPt \triangle \partial(\sigma))\cap{X_{t}}}  & \sigma \in \QQt . \\
           \end{cases}
\end{equation*}
If the introduced vertex is not a $d+1$-simplex, then it is $d$-simplex $\rho$ so that $X_t = X_{s} \cup \{\rho\}$. If $\rho \in \PPt \triangle \partial(\QQt ) \triangle V$ then $\tab{t,\QQt ,\PPt } =\infty$. Otherwise 
\begin{equation*}
    \tab{t,\QQt ,\PPt } =
     \begin{cases}
               	\tab{s,\QQt  ,\PPt \setminus\{\rho\}} & \rho \in \PPt \\
               	\tab{s,\QQt ,\PPt }          	& \rho \notin \PPt .   \\
           \end{cases}
\end{equation*}
\newline

\textbf{Forget Node}:
\newline
Let $t$ be a forget node with child $s$. There are two cases. Either we forget a vertex corresponding to a $d$-simplex $\rho$ and then $ \tab{t,\QQt ,\PPt } = \min \{ \tab{s,\QQt ,\PPt }, \tab{s,\QQt  ,\PPt \cup \{\rho\}} + \cosst(\rho)\}.$ If not, we forget a $d+1$-simplex $\sigma$ and $\tab{t,\QQt ,\PPt } = \min \{ \tab{s,\QQt ,\PPt }, \tab{s,\QQt  \cup \{\sigma\} ,\PPt }\}.$
\newline

\textbf{Join node}: 
\newline
Finally, we consider the case when $t$ is join node with two child nodes $s$ and $s'$ and we have $X_t=X_{s}=X_{s'}$. Let $\QQt $ be a subset of the $d+1$-simplices of $X_t$ and $\PPt $ be a subset of the $d$-simplices of ${X_t}$ as before. Let $({X_t})_d$ denote the subset of $d$-simplices of ${X_t}$ and compute
$$\tab{t,\QQt ,\PPt } =\min_{\Ps, \Pss \subseteq {({X_t})}_d}\big\{\tab{s,\QQt ,\Ps }+\tab{s',\QQt ,\Pss } \Big | \PPt  = (\Ps \triangle \Pss \triangle \partial(\QQt ) \triangle V)\cap {X_t}\big\}.$$

\begin{figure}[h!]
\centering
\includegraphics[width=130mm]{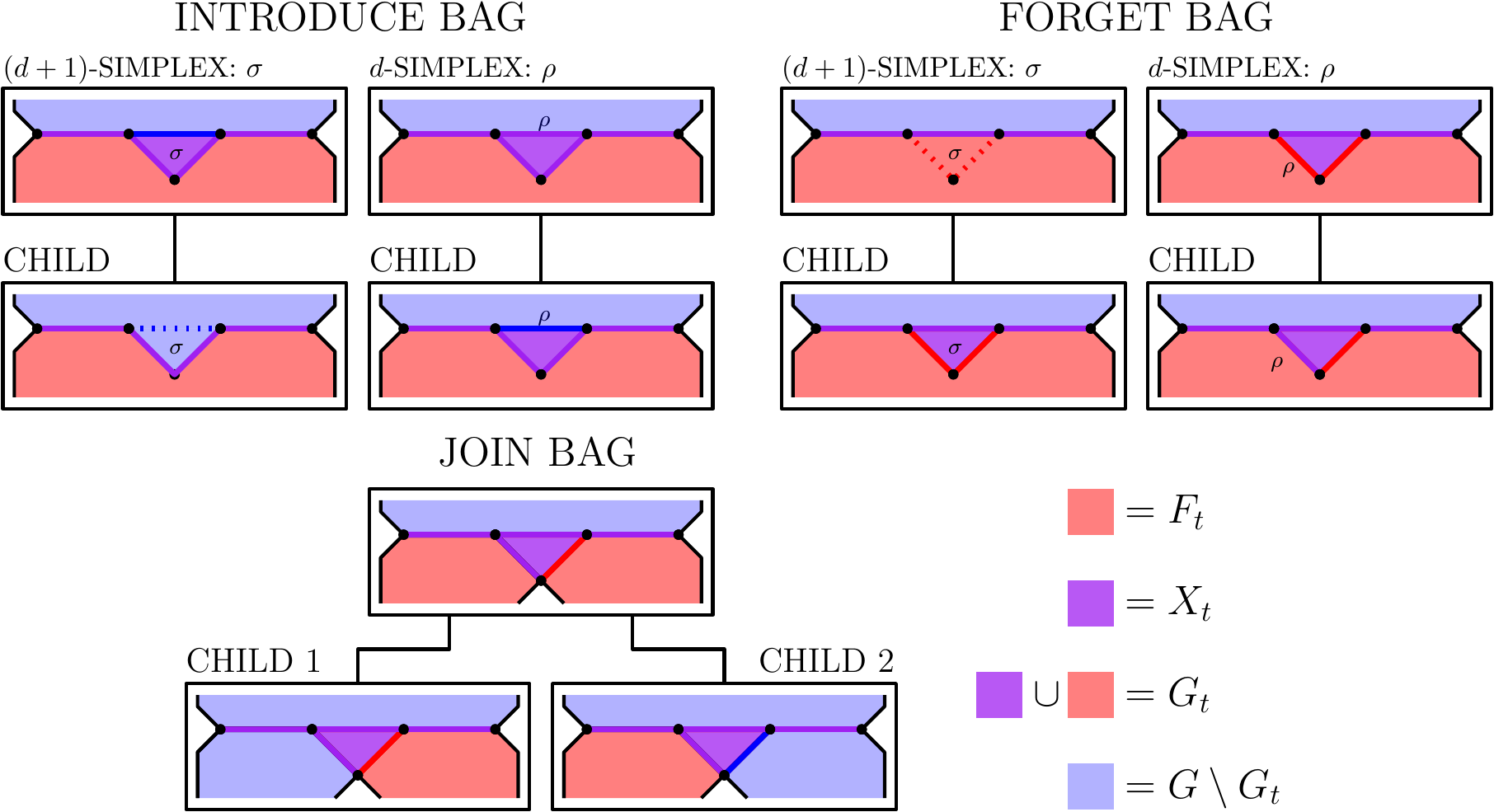}
\caption{\label{FIG: ALG2 aid} The analogue to \cref{FIG: ALG1 aid} for the Hasse diagram based algorithm}
\end{figure}

\subsubsection{Correctness}
Let $\mathcal{S}[t,\Qt ,\Pt ]$ be the set of solutions to the R-HL$_d$ on input tuple $(G_t,X_t,\Qt ,\Pt )$ and let its elements be pairs consisting of a $d+1$-chain and a $d$-cycle, $(W,U)$. Recall that pairs $(W,U)$ have the following properties: $U$ is homologous to $V$ through $W$, i.e. $U= V \triangle \partial(W)$, $W$ intersects $X_t$ at $\Qt $ and $U$ intersects $X_t$ at $\Pt $.

The main difference between the two algorithms is how much the problem changes at each introduce/forget bag. We had process all the $d$-simplices at the boundary of an introduced $d+1$-simplex simultaneously in the connectivity based algorithm. This time we can get to take care of the $d$-simplices individually in its own bag. The proof of correctness is virtually the same and when they differ it is simpler this time around since fewer things are happening at the same time. We outline which modifications are needed in order to obtain a proof of correctness at the introduce node. 

There are now two cases and we deal with them separately according to the dimension of the introduced simplex. First assume we introduce a simplex $\sigma$ of dimension $d+1$. Like in the previous algorithm all we have to prove is that $\sigma \not \in \Qt$ then $(W,U)\in \mathcal{S}[t,\Qt ,\Pt ]$ if and only if  $(W,U)\in \mathcal{S}[s,\Qt ,\Pt ]$ and that if $\sigma \in \Qt$ then $(W,U)\in \mathcal{S}[t,\Qt ,\Pt ]$ if and only if  $(W\setminus \{\sigma\},U\triangle \partial (\sigma))\in \mathcal{S}[s,\QQt \setminus\{\sigma\} ,(\PPt \triangle \partial(\sigma))\cap{X_{t}}]$. 

Next let us assume that we introduce the simplex $\rho$ of dimension $d$. First we note that if $\rho \in \PPt \triangle \partial(\QQt ) \triangle V$ then there can be no solution to $c[t,\Qt ,\Pt ]$. This is because no $d+1$-simplex from $F_t$ is adjacent to $\rho $ by the properties of a tree decomposition. Thus the only simplices in a solution $W$ that are cofaces of $\rho$ are those in $X_t$ but these simplices are precisely those in $\Qt$. Making the observations that $(W,U)\in \mathcal{S}[t,\Qt ,\Pt ]$ if and only if $(W,U)\in \mathcal{S}[s,\Qt ,\Pt\setminus \{\rho\} ]$ when $\rho \in P_t$ and that  $(W,U)\in \mathcal{S}[t,\Qt ,\Pt ]$ if and only if  $(W,U)\in \mathcal{S}[s,\Qt ,\Pt ]$ when $\rho \not \in P_t$ completes the proof.

\subsubsection{Runtime Analysis}
The analysis is similar to the previous one but we get different numbers.  The leaf nodes still take constant time to process. At introduce and forget bags we need to compute and store $2^{k}$ solutions each of which can be found in constant time meaning they can be processed in $\mathcal{O}(2^{k})$ time. The join bag also needs to store and compute up to $2^{k}$ problems but since there are also up to $2^{k}$ pairs of sets $\Ps$ and $\Pss $ so that $\Pt  = \Ps\triangle \Pss \triangle \partial(\Qt)\triangle (V\cap {X_t})$ each computation needs $\mathcal{O}(2^{k})$ time. This means that the join bag can be computed in $\mathcal{O}(2^{k}\cdot 2^{k}) = \mathcal{O}(2^{2k})$ time and that the algorithm needs $\mathcal{O}(2^{2k}n) = \mathcal{O}(4^{k}n)$ time.

\section{Optimality under the ETH}
\label{sec:eth}

We prove that both our treewidth based algorithms are optimal up to the base of the exponent if we assume that the exponential time hypothesis (ETH) is true.

\subsection{The Exponential Time Hypothesis}
Let $n$ be the number of variables in a $3$-SAT formula.

\begin{definition}[ETH] $3$-SAT cannot be solved in $2^{o(n)}$-time.
\end{definition}
This is a common formulation of the ETH which is slightly weaker than the original hypothesis. We say that a parameterized algorithm is \emph{ETH-tight} (with respect to the parameter $k$) if it runs in $2^{\mathcal{O}(f(k))}n^c$-time and solving it in $2^{o(f(k))}n^c$-time contradicts the ETH.

A \emph{cut} in a graph is a partitioning of the vertices into two subsets and the \emph{size of a cut} is the number of edges going between vertices on opposite sides of the partition.

\begin{definition} \textsc{Max Cut}\\
Input: A graph $G$, a tree decomposition $X$ of $G$ and an integer $k$.\\
Output: YES if there is a cut in $G$ of size $k$ or greater, NO otherwise.
\end{definition}

Let $k$ be the treewidth of the graph given as input to \textsc{Max Cut}.

\begin{theorem}[\cite{lokshtanov2011known}]
\textsc{Max Cut} cannot be solved in $2^{o(k)}n^c$-time if the ETH is true.
\end{theorem}

\subsection{ETH Tightness}

\begin{theorem}\label{Theorem: Optimality}
The HL$_1$ problem cannot be solved in $2^{o(k)}n^c$-time if the ETH is true even when the input is restricted to $2$-manifolds which can be embedded in $\mathbb{R}^3$.
\end{theorem}

\begin{proof}
We prove this by reducing \textsc{Max Cut} to the HL$_1$ problem where the graph $G$ is mapped to a $2$-manifold  $Y(G)$. This space is constructed by first associating a $2$-sphere $Y^u$ to every vertex $u\in G$. Then for every edge $uv \in G$ we glue $Y^u$ and $Y^v$ together by removing an open disc from both $Y^u$ and $Y^v$ and identifying the boundaries. Any two discs removed from $Y^u$ should have non-intersecting boundaries, see \cref{FIG: ETH reduction basic idea}.  

\begin{figure}[h!]
\centering
\includegraphics[width=160mm]{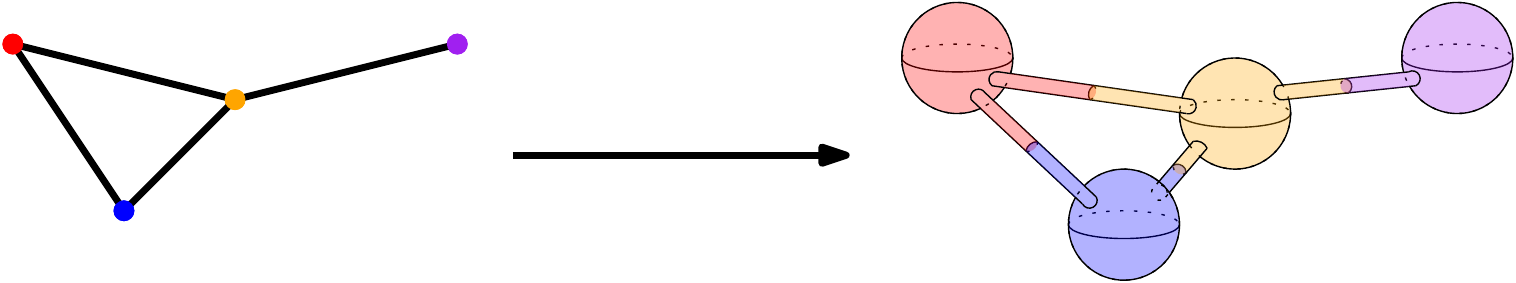}
\caption{\label{FIG: ETH reduction basic idea} $Y(G)$ can be thought of as a the surface of a fattened version of the graph $G$.}
\end{figure}

To complete the reduction we give $Y(G)$ the structure of a simplicial complex with the following properties.
\begin{itemize}
    \item Every $2$-simplex of this triangulation is contained in the interior of exactly one $Y^u$.
    \item Every $1$-simplex is either in the interior of some $Y^u$, in which case we give it the weight $\infty$, or in the intersection $Y^u\cap Y^v$, in which case it gets the weight $1/||(Y^u\cap Y^v)_1||$.
\end{itemize} 
In \cref{LEMMA: triangulation lemma} below, we show that this can be done such that the treewidth of the connectivity graph of $Y(G)$ grows at most linearly in the treewidth of $G$. Finally we set $V = \cup_{uv\in G}(Y^u\cap Y^v)_1$. Since we know that $Con_2(Y(G))$ is connected, any solution $U = V \triangle \partial(W)$ which is optimal must be of the form $W = \cup_{x\in X}Y^x_2$ for some subset $X$ of vertices in $G$. It is then easy to verify that a we have a max cut in $G$ with right side $X$ if and only if the shortest cycle homologous to $V$ in $Y(G)$ is $U = V \triangle \partial(\cup_{x\in X}Y^x_2)$. 
\end{proof}

The idea behind \cref{LEMMA: triangulation lemma} is simple but the proof is a bit technical. We have therefore added \cref{FIG: Tree decomposition nice ETH}, an example showing how a nice tree decomposition of the graph in \cref{FIG: Tree decomposition} can be used to make a space homeomorphic to $Y(G)$. We hope that it is clear that the connectivity graph/Hasse diagram of this space can be given a triangulation of low treewidth. The figure represents a 2-dimensional surface (see \cref{FIG: ETH reduction components} for details) without a boundary. It is is made by cutting and gluing together stretched $2$-spheres, 

\begin{figure}[!h]
\centering
\begin{subfigure}{.5\textwidth}
  \centering
  \includegraphics[width=0.98\linewidth]{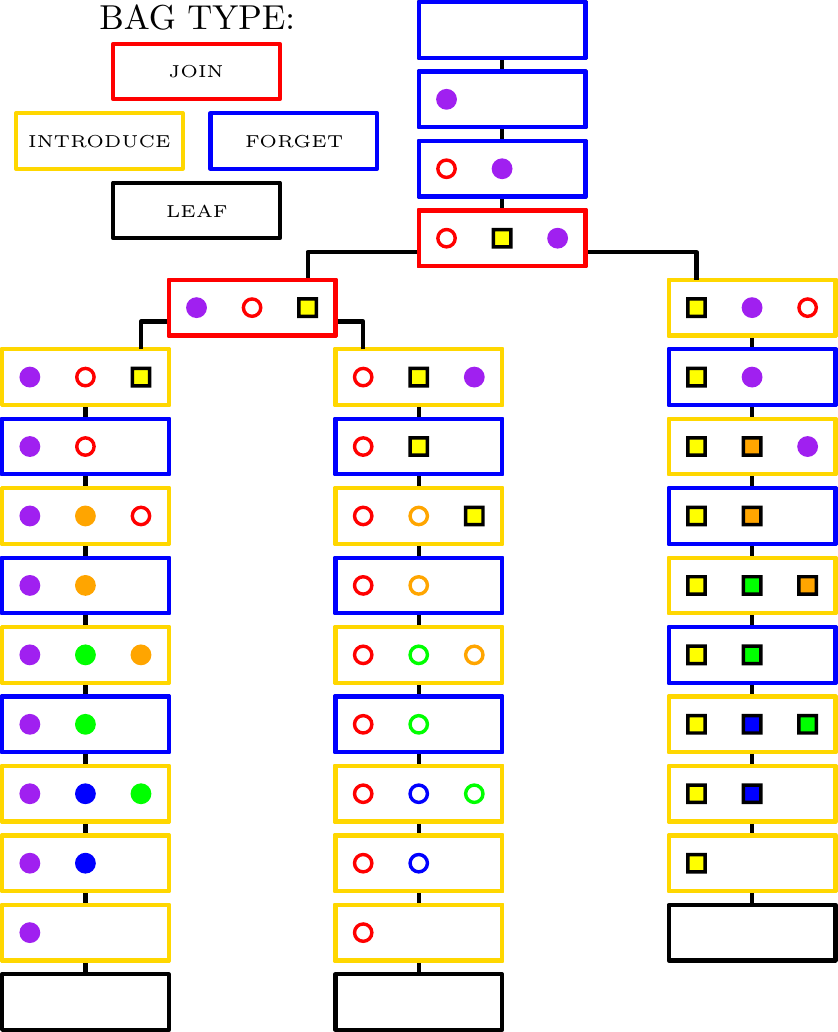}
  \caption{A nice tree decomposition of $G$}
\end{subfigure}%
\begin{subfigure}{.5\textwidth}
  \centering
  \includegraphics[width=0.75\linewidth]{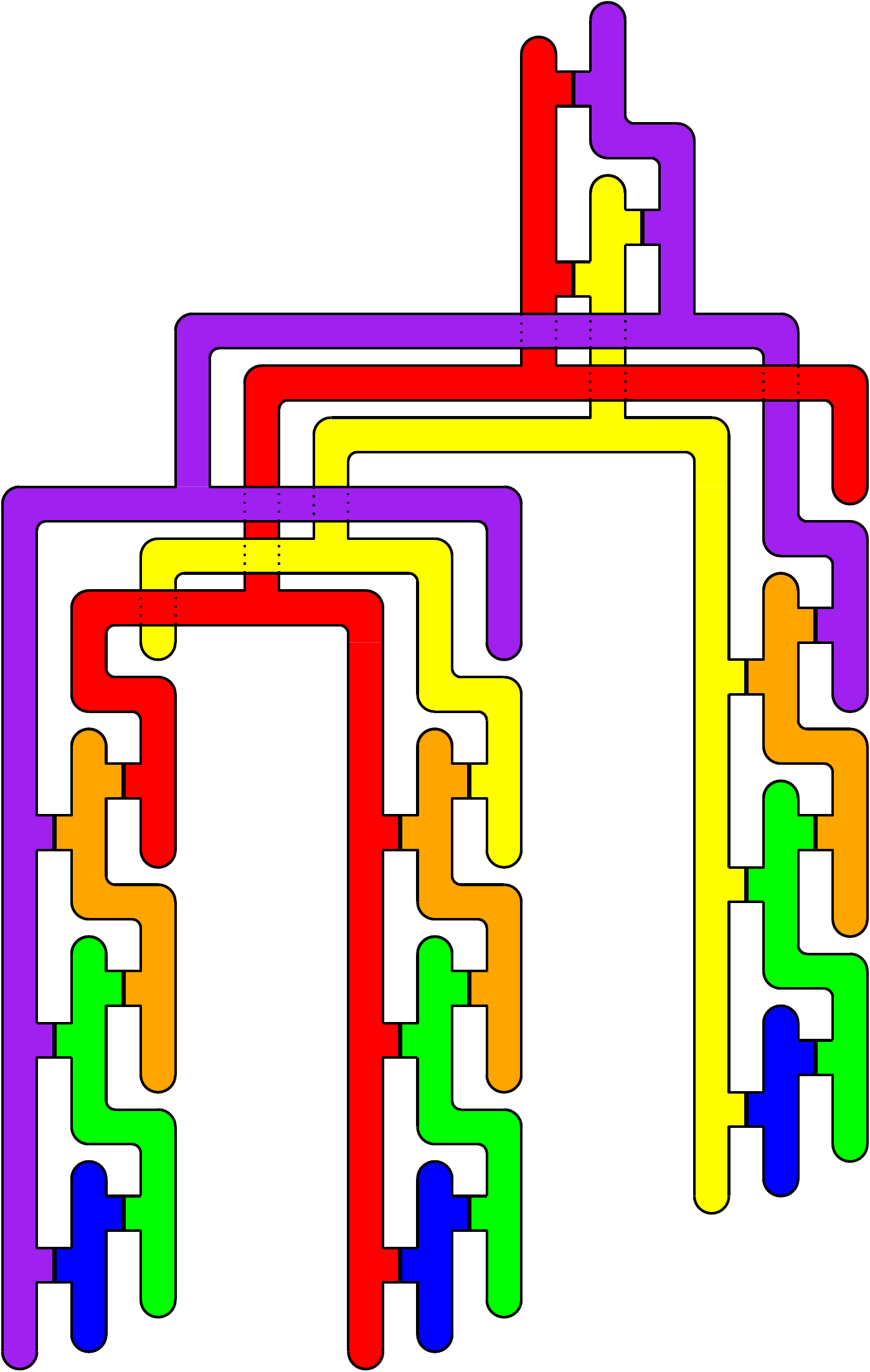}
  \caption{The corresponding space $Y(G)$.}
\end{subfigure}
\caption{\label{FIG: Tree decomposition nice ETH} A nice tree decomposition of the graph in \cref{FIG: Tree decomposition} (left) and a picture of the space it would have been mapped to under the above reduction (right).}
\end{figure}

\begin{lemma} \label{LEMMA: triangulation lemma}
There is a constant $c\in\mathbb{N}$ so that for any graph $G$ we can triangulate $Y(G)$ so that $tw(Con_2(Y(G))) \leq c \cdot tw(G)$ and every $Y^u$ is a subcomplex with connected $Con_2(Y^u)$.
\end{lemma}

\begin{proof}
Let  $TD(G)$ denote a nice tree decomposition of $G$ of width $k$. We will use structural induction on $TD(G)$  to construct a triangulation of $Y(G)$. Given a bag $X_t$ in $TD(G)$ we assume that we have already constructed a simplicial complex $Y_t$ which is also a $2$-manifold with one boundary component for each vertex in $X_t$. We also assume that we have found a tree decomposition of $Con_2(Y_t)$, denoted by $TD(Y_t)$, and that $Y_t$ and $TD(Y_t)$ satisfy the following properties:
\begin{enumerate}
    \itemsep0pt
    \item The width of $TD(Y_t)$ is bounded by $c\cdot tw(G)$.
    \item The space $Y_t$ is the union of subcomplexes $Y_t^u$ where $u$ is a vertex of $G$ in either $F_t$ or $X_t$. Each $Y_t^u$ is a $2$-sphere that has $||N(u)||$ holes if $u\in F_t$ and $||N(u)\cap F_t|| + 1$ holes if $u\in X_t$, where $N(u)$ is the set of vertices adjacent to $u$.
    \item Two components $Y_t^u$ and $Y_t^v$ intersect each other at precisely one common boundary if and only if $uv\in G_t$ and at least one of $u$ or $v$ is in $F_t$. 
    \item $TD(Y_t)$ contains a distinct bag $R_t$ containing precisely the $2$-simplices of ``cylinders'' $C_t^u \subseteq Y_t^u$ for $u\in X_t$ where each cylinder contain one boundary component of $Y_t$.
\end{enumerate}

At the root bag every node of $G$ is in $F_t$. This means that $Y_r$ is a simplical complex homeomorphic to $Y(G)$ having a tree decomposition $TD(Y_r)$ with width at most $c\cdot tw(G)$. The leaf bags $X_t$ are the base cases. We set $Y_t= \emptyset$ and $TD(Y_t)$ to be the tree decomposition containing just an empty bag $R_t$. The construction trivially fits the above requirements. We will show how to construct $Y_t$ and $TD(Y_t)$ for the other kinds of bags visually, using the notation presented in \cref{FIG: ETH reduction components}.

\begin{figure}[h!]
\centering
\includegraphics[width=120mm]{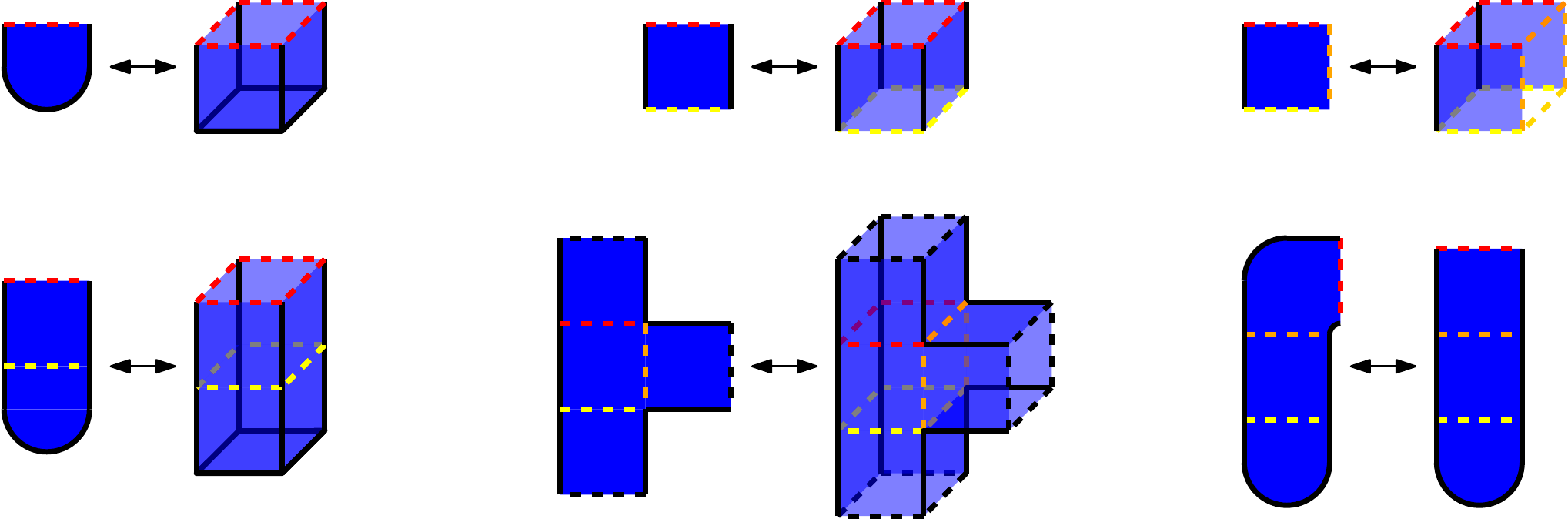}
\caption{\label{FIG: ETH reduction components} The basic building blocks used to build spaces. The figures shown to the left represents the figure shown to the right. Each square in the $3D$-figure consists of two $2$-simplices glued along a common face. Top left to top right we see a box without its top (i.e. a disc), a cylinder and a cylinder without one of its sides. The bottom row of figures shows how the above components can be glued together.}
\end{figure}

Let $X_t$ be an introduce bag in $TD(G)$ with the child bag $X_{t'} = X_t\setminus \{v\}$. By the induction hypothesis we have already constructed a space $Y_{t'}$ and a tree decomposition $TD(Y_{t'})$ with a distinct bag $R_{t'}$. The constriction of $Y_t$, $TD(Y_t)$ and $R_t$ is shown in \cref{FIG: ETH reduction introduce}. Clearly $TD(Y_{t})$ is a tree decomposition. The new bag's we have added is only a constant times bigger than $X_t$, giving us Property 1. The new $2$-connectivity graph of each $Y^v$ is connected and each $Y^v$ containing as many holes as Property 2 requires. The $Y^v$'s intersects as demanded by Property 3 since the new vertex cannot be a neighbour of a node that has been forgotten earlier. Property 4 is satisfied by construction.

\begin{figure}[!h]
\centering
\begin{subfigure}{.3\textwidth}
\centering
  \includegraphics[width=0.9\linewidth]{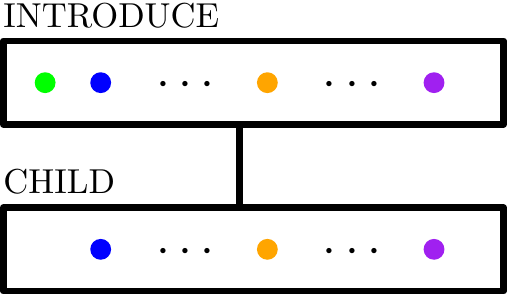}
\end{subfigure}%
\begin{subfigure}{.33\textwidth}
  \centering
  \includegraphics[width=0.67\linewidth]{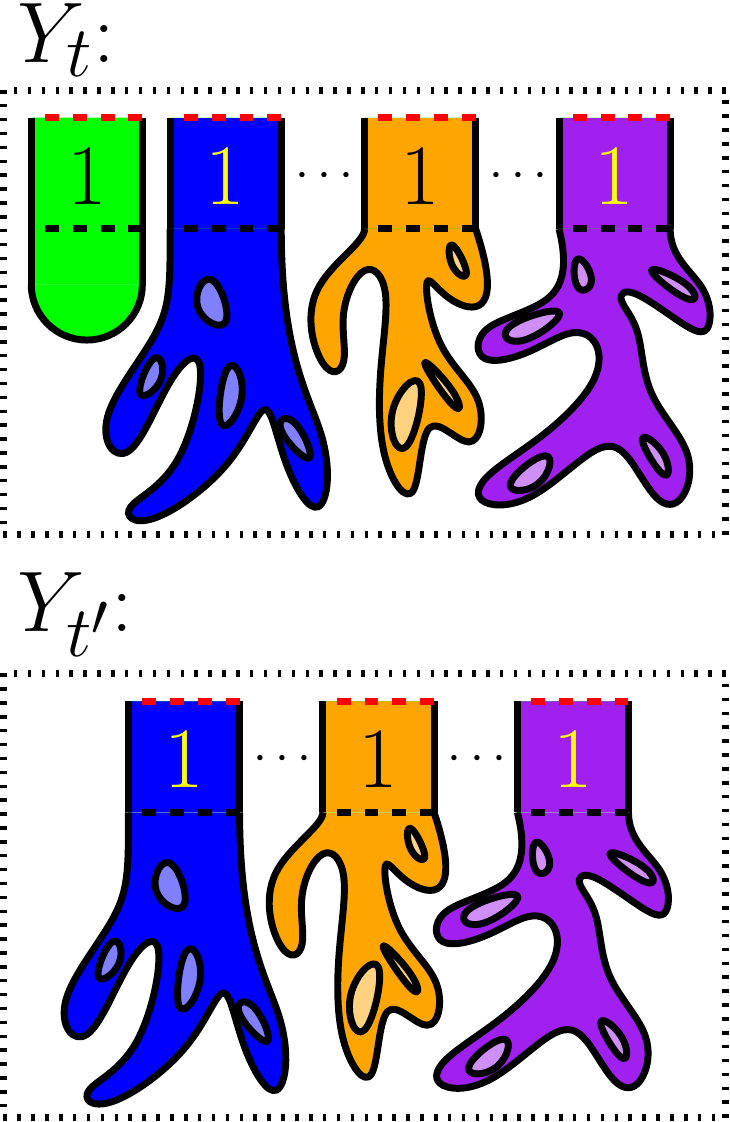}
\end{subfigure}
\begin{subfigure}{.33\textwidth}
\centering
\includegraphics[width=0.725\linewidth]{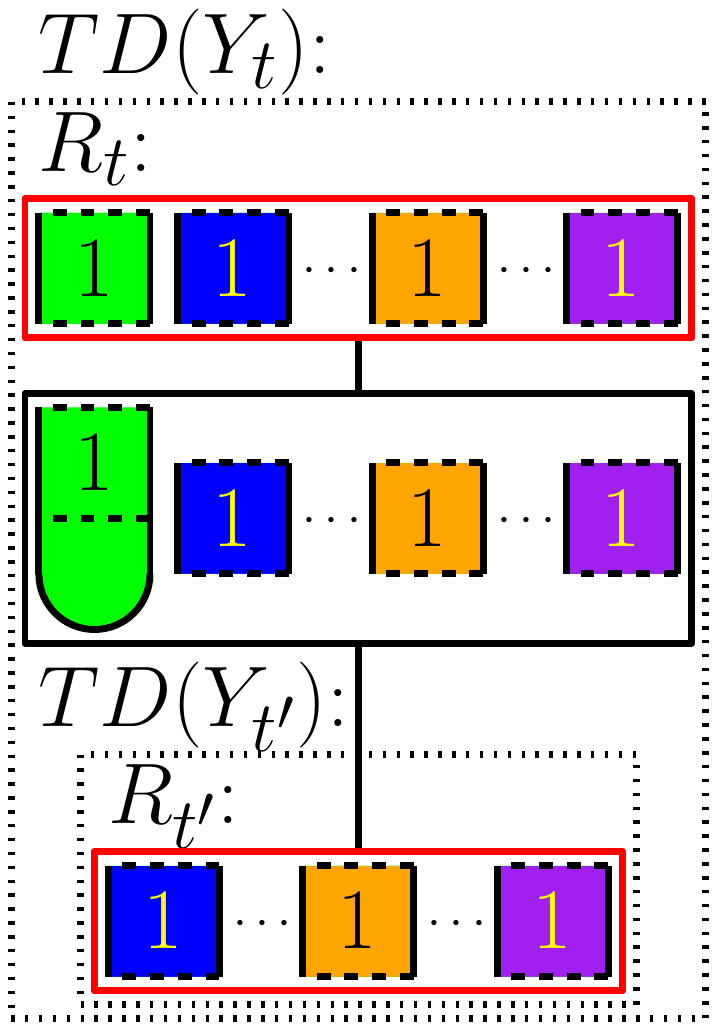}
\end{subfigure}
\caption{\label{FIG: ETH reduction introduce} A ``green'' vertex, $v$, is introduced, so we add a green disc to $Y_{t'}$ to make $Y_t$. A tree decomposition of this space is constructed by adding two bags. The first is adjacent to $R_{t'}$ and contains the simplices of $R_{t'}$ as well as all the new $2$-simplices. $R_t$ is a subset of this bag and it is adjacent to this bag in $TD(Y_t)$.}
\end{figure}

Next we consider the forget bag $X_t$ and its child bag $X_{t'} = X_t\cup \{v\}$. This is the stage where edges of $G$ are ``accounted'' for. 
The space and the tree decomposition is described visually in \cref{FIG: ETH reduction forget}. That Property 1-4 holds is obvious as we know that for every edge $uv$ in $G$ there is exactly one bag containing $u$ and $v$ where either $u$ or $v$ is forgotten.

\begin{figure}[!h]
  \centering
\begin{subfigure}{.3\textwidth}
  \centering
  \includegraphics[width=0.9\linewidth]{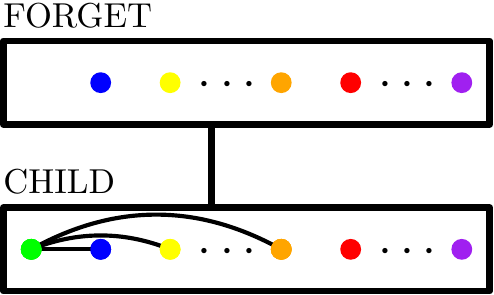}
\end{subfigure}%
\begin{subfigure}{.33\textwidth}
  \centering
  \includegraphics[width=0.88\linewidth]{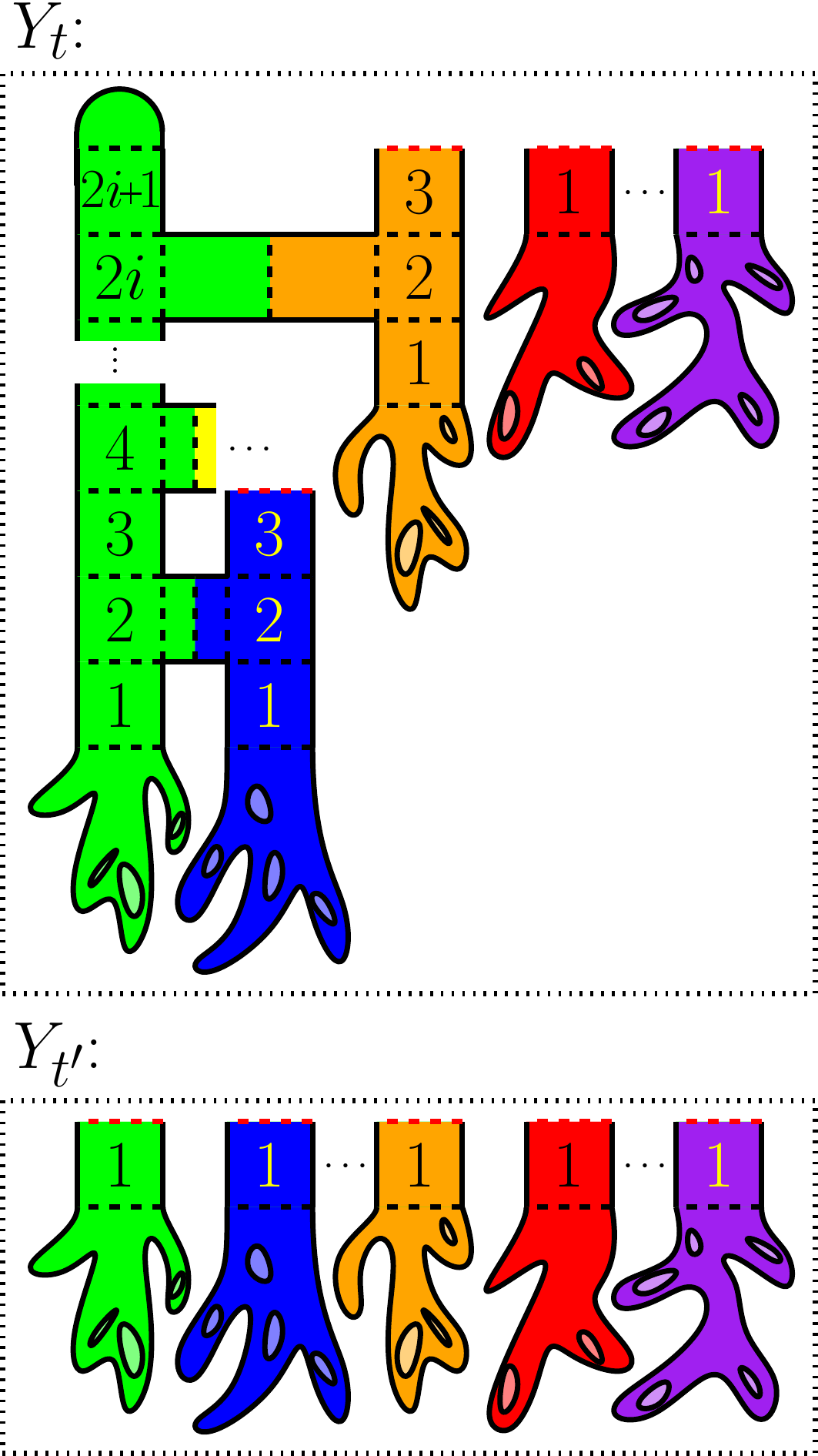}
\end{subfigure}%
\begin{subfigure}{.33\textwidth}
  \centering
  \includegraphics[width=0.975\linewidth]{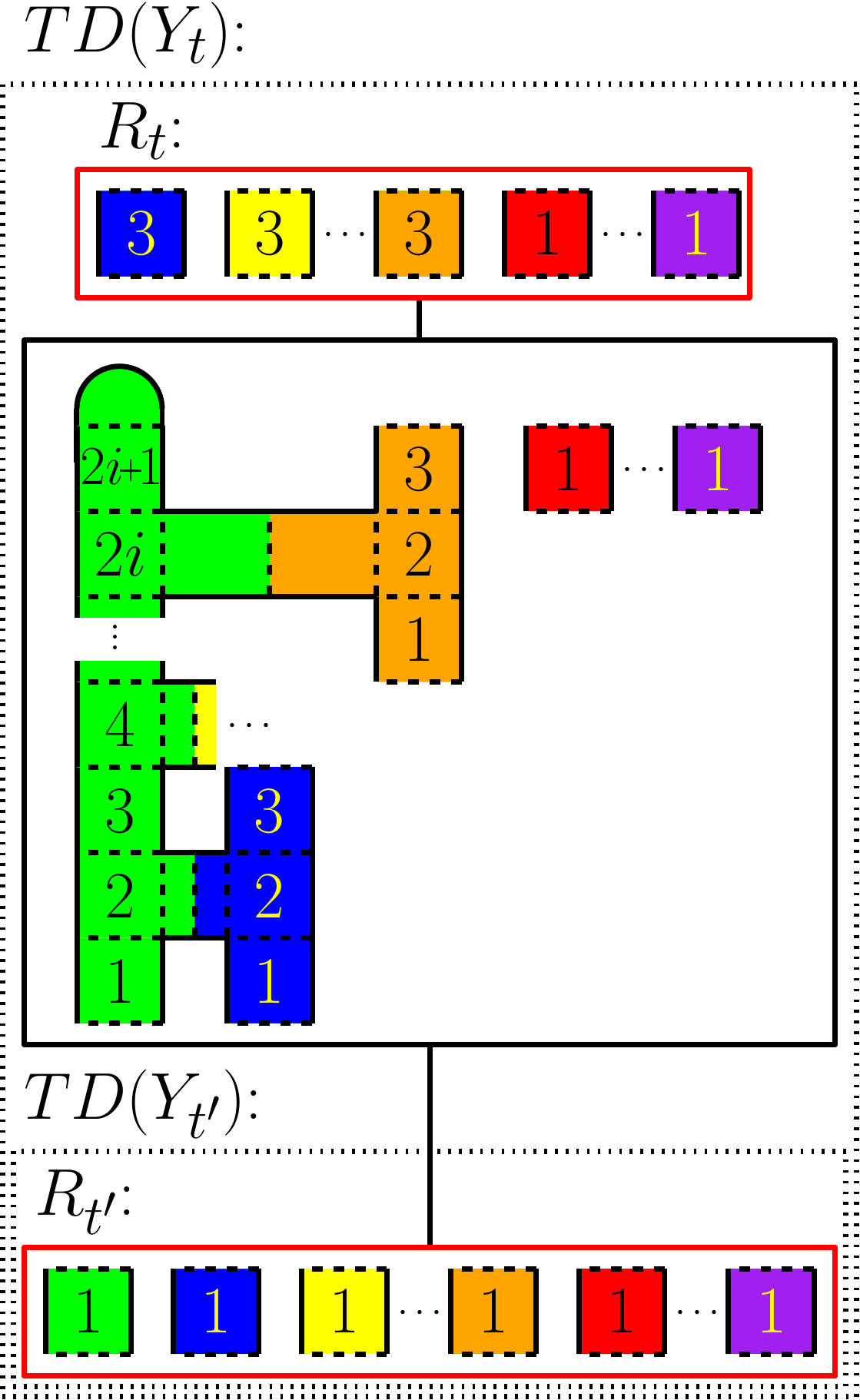}
\end{subfigure}
    \caption{\label{FIG: ETH reduction forget} A green vertex $v$ is forgotten. First we attach it to the components that corresponds to its neighbours in $X_t$ before we glue a disc to its ``top''. We get a tree decomposition by adding a bag adjacent to $R_{t'}$ containing the new $2$-simplices and the content of $R_{t'}$. $R_t$ is a subset of this bag and is adjacent to this bag in $TD(Y_t)$.}
\end{figure}

Finally we consider what happens at the join bags $X_t$ in $TD(G)$ which have two child bags, $X_{t'}$ and $X_{t''}$, both equal to $X_t$. For each vertex $v\in X_t$ we construct $Y_t^v$ from $Y_{t'}^v$ and $Y_{t''}^v$ by connecting them as pictured in \cref{FIG: ETH reduction join}. The new tree decomposition is also shown in this figure. Verifying that Properties $1$-$4$ are satisfied is again elementary.
\begin{figure}[!h]
  \centering
  \includegraphics[width=0.4\linewidth]{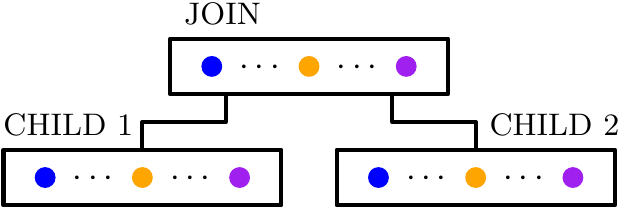}
\centering
\begin{subfigure}{.5\textwidth}
  \centering
  \includegraphics[width=0.84\linewidth]{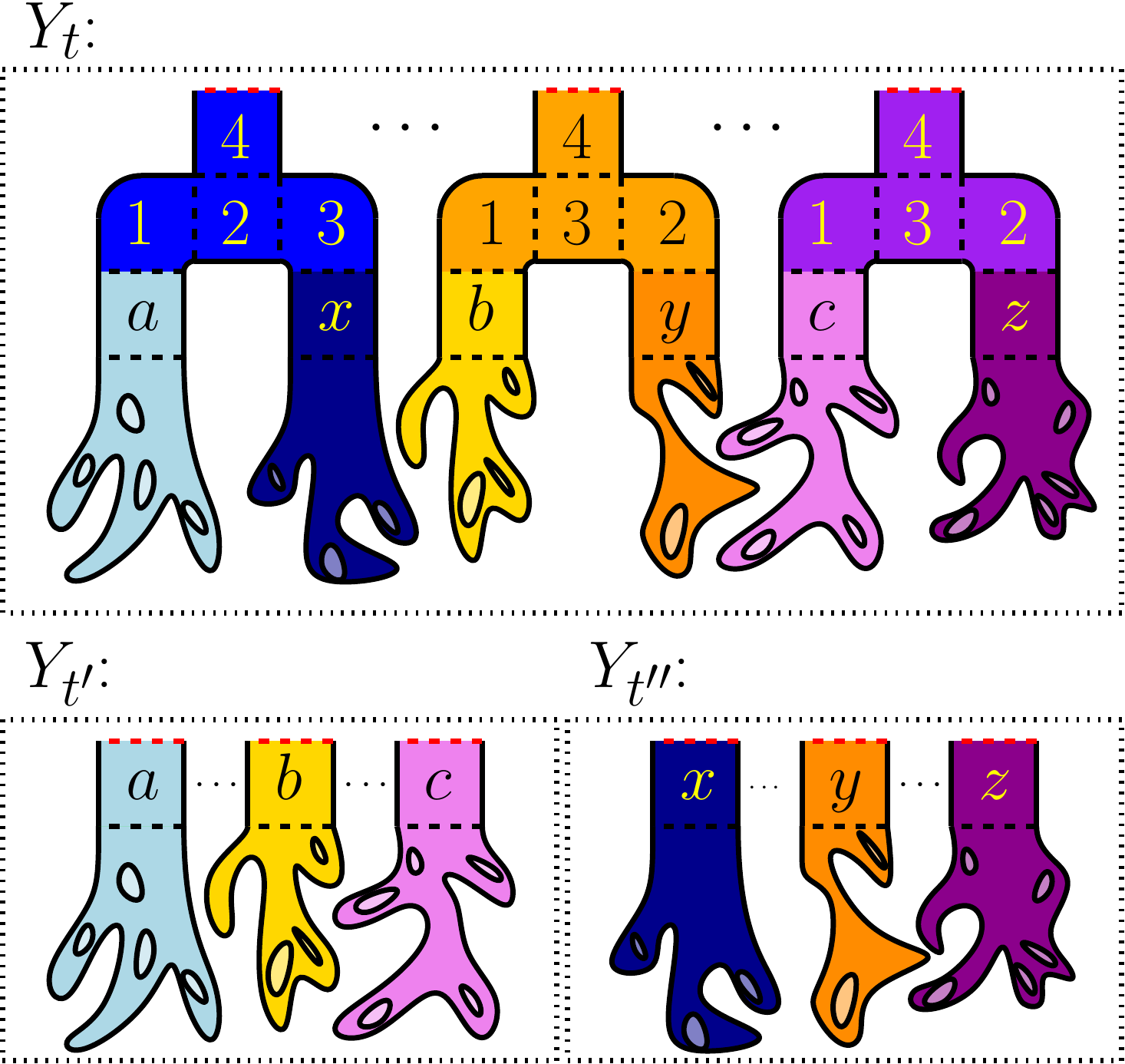}
\end{subfigure}%
\begin{subfigure}{.5\textwidth}
  \centering
  \includegraphics[width=0.9\linewidth]{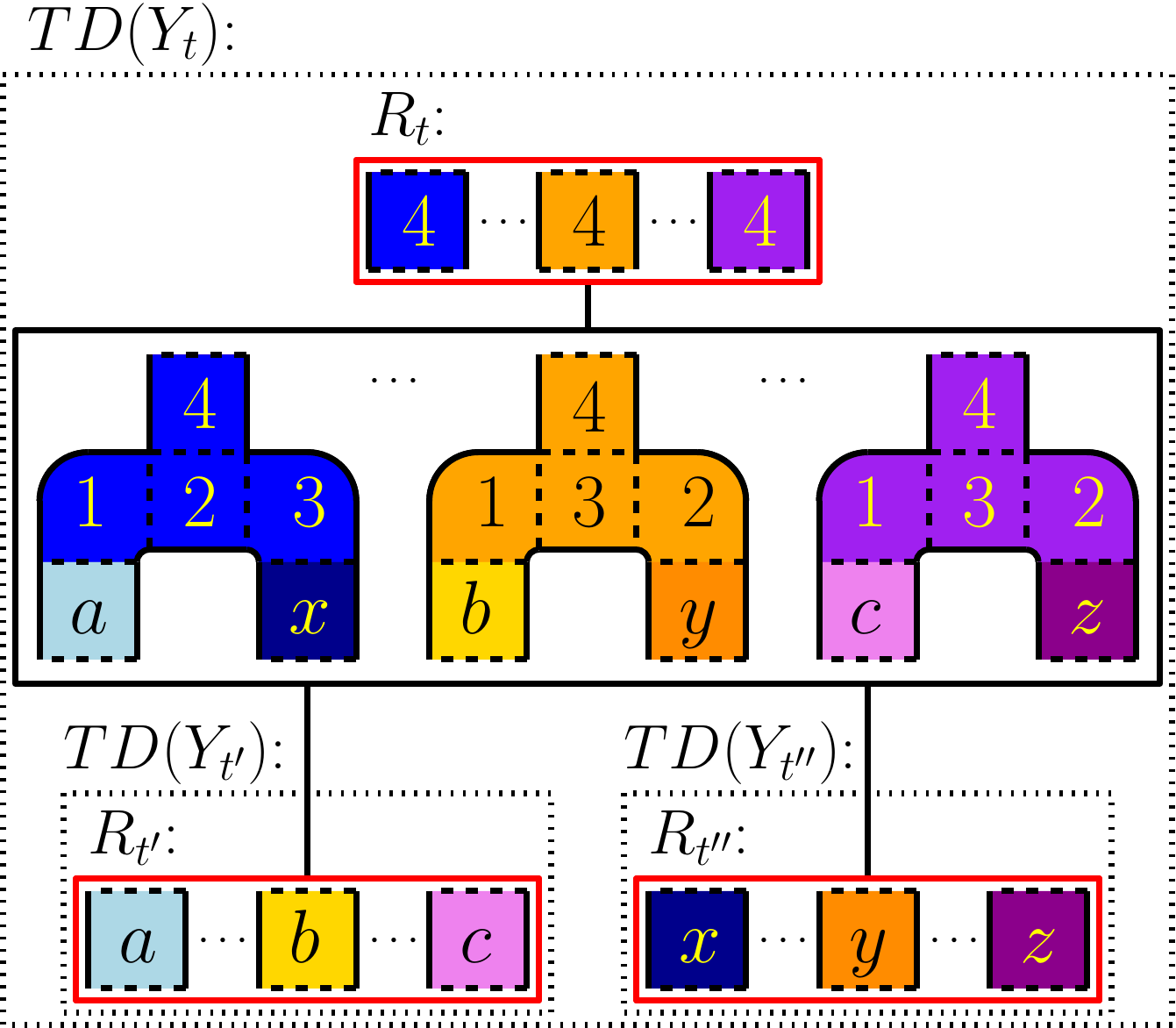}
\end{subfigure}
    \caption{\label{FIG: ETH reduction join} We have two components for every vertex $v \in X_t$, one for each child bag. We combine these into one component by gluing them to a sphere with three holes. To get the bag $R_t$ we attach a cylinder to the last hole.  }
\end{figure}
\end{proof}



\cref{Theorem: Optimality} holds in higher dimensions as well. In particular if $d \geq 1$ and $k$ is the treewidth of either $Con_{d+1}(K)$ or $Hasse_{d+1}(K)$ we have the following corollary:

\begin{cor}
The HL$_d$ problem cannot be solved in $2^{o(k)}n^c$-time if the ETH is true. This is still the case if the input is restricted to to $d+1$-manifolds embedded in $\mathbb{R}^{d+2}$
\end{cor}

\begin{proof}
First we note that the space constructed in \cref{LEMMA: triangulation lemma} is an orientable $2$-manifold so it embeds in $\mathbb{R}^3$. The rest of the proof is by induction using \cref{Theorem: Optimality} as a base case. We reduce from the HL$_{d}$ problem to the HL$_{d+1}$ by using suspension as discussed in \cref{sec:preliminaries}. It remains to prove the general fact that $tw(Con_{d}(S(K)) \leq 2tw(Con_{d+1}(K))$. If we have a tree decomposition of $Con_{d}(K)$ then a valid tree decomposition of $Con_{d+1}(S(K))$ is obtained by replacing every bag $X_t$ in this tree decomposition with the bag $X'_t = \{ \sigma \cup  v^+, \sigma v_- \cup  | \sigma \in X_t \}$. It is easy to see that this is again a tree decomposition. The argument for the Hasse diagram is analogous.
\end{proof}

\section{Experiments}
\label{sec:comparison}

We could not determine in advance which of the treewidth based algorithms would work best in practice as they have different advantages when compared against each other. In order to settle this we conducted experiments on Python implementation of the two algorithms. This section contains a report on the result of these experiments as well as some observations on the theoretical differences between the algorithms. The code we used to solve homology localization is now freely available at \url{https://github.com/erlraavaa314/homology-localization}.

\subsection{Implementation}

This section gives some further information about our implementation. 

\subsubsection{Finding a Tree Decomposition}
Prior to running our algorithm we needed to find nice tree decompositions. We used two different heuristic algorithms from \texttt{networkx} to do this. One is based on the Minimum Fill-in heuristic and the other is based on the Minimum Degree heuristic. We used the decomposition that had the smallest width. This does not mean that we got optimal tree decompositions but they seemed to be reasonably good.

\subsubsection{Key Adjustments}
To speed up the computation we made some minor adjustments in the implementation of the algorithm. The most substantial change was made in order to avoid storing and computing solutions to subproblems that were obviously impossible to solve. This was done by keeping a nested dictionary where only entries for problems with feasible solutions were solved. To make this work we had to fill the tables by iterating through the solutions stored at the child bag. Temporary solutions were stored as entries at the parent bag.

\subsubsection{Problem Instances}

We timed our code on seven different kinds of spaces which we describe below. The first four of these are illustrated in \cref{FIG: Four kinds of spaces} 
\begin{itemize}
    \itemsep0pt
    \item Rectangular triangulated surfaces based on vertices placed in a grid.
    \item Cylinders obtained by gluing a pair of opposite sides of a rectangular surface.
    \item Tori obtained by gluing both pairs of opposite sides of a rectangular surface 
    \item ``M-spaces'' made by gluing $k$ rectangular $2$ by $m+1$ surfaces to a circle on $m$ vertices. We attach the short side of rectangle $i$ to the ``opposite'' short side of rectangle $i+t \pmod k$. 

    \item The Vietoris-Rips complex of three kinds of point clouds. These simplicial complexes were weighted by the length of the $1$-simplices.
    \begin{enumerate}
    \itemsep0pt
        \item Between $5$ and $20$ points are sampled uniformly at random from a unit circle. These point are then multiplied by a scalar chosen using a normal distribution with expectation $1$ and standard deviation $0.1$. Finally, each point is given a normally distributed $z$-coordinate expectation $0$ and standard deviation $0.1$. Referred to later as Unfiltered (P). 
        \item Between $10$ and $200$ points are sampled in the same way as above. Then we apply the following procedure. First we pick a point at random, then we select the point furthest away from the set of points we currently had until half of the points are chosen. Referred to later as Filtered (P). 
        \item We pick points on a circle by first dividing it into arc segments of the same length. Points are selected uniformly at random from each arc. A random radius and a random $z$ coordinate were then assigned to each point like before.  We tested different numbers of arc segments, from $20$ to $100$, and different number of points from each segment, one, two and three. This means we selected between $20$ and $300$ points. These spaces are referred to later as Sector (P).
    \end{enumerate}
\end{itemize}

\begin{figure}[h!]
\centering
\hfill
  \begin{subfigure}[b]{0.365\textwidth}
  \centering
    \includegraphics[width=\textwidth]{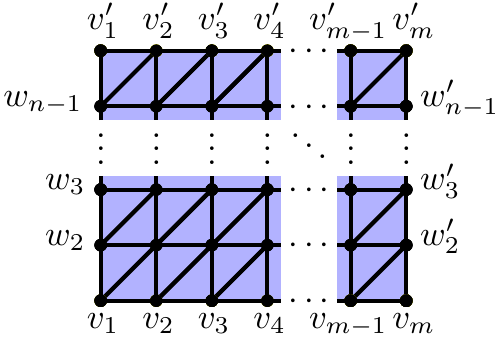}
    \caption{Rectangle.}
  \end{subfigure}
\hfill
  \begin{subfigure}[b]{0.31\textwidth}
  \centering
    \includegraphics[width=\textwidth] {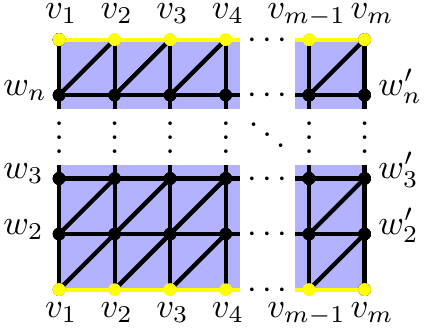}
    \caption{Cylinder.}
  \end{subfigure}
\hfill
 \begin{subfigure}[b]{0.31\textwidth}
    \centering
    \includegraphics[width=\textwidth] {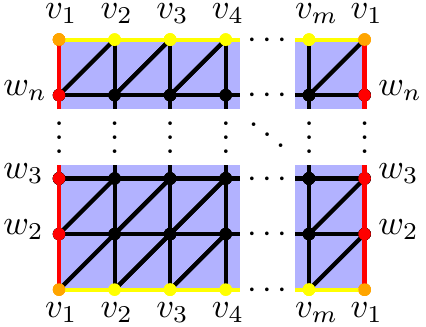}
    \caption{Torus.}
  \end{subfigure}
\centering
    \begin{subfigure}[b]{0.4\textwidth}
    \centering
    \includegraphics[width=\textwidth]{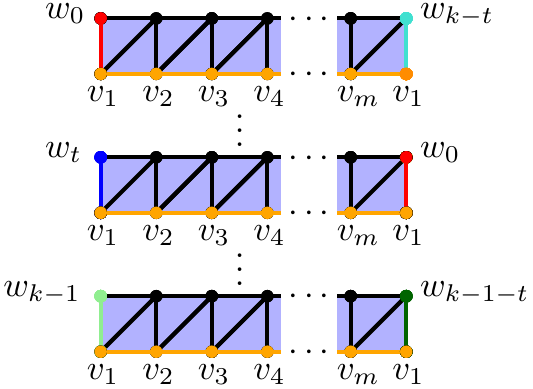}
    \caption{M-space.\label{fig: Moebiusish}}
  \end{subfigure}
\centering
    \begin{subfigure}[b]{0.35\textwidth}
    \centering
    \includegraphics[width=\textwidth]{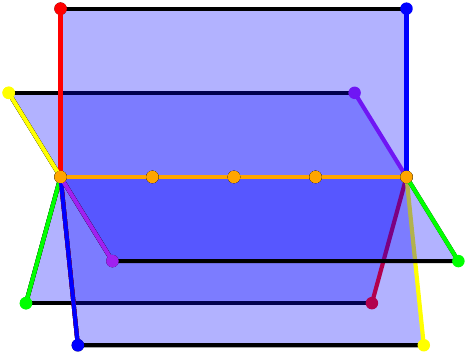}
    \caption{M-space: $m =4$, $k = 5$, $t = 2$.}
  \end{subfigure}
  
 \caption{\label{FIG: Four kinds of spaces}Four out of the seven kinds of simplicial complexes that we used as input in order to compare the efficiency of the algorithms.}
\end{figure}

We used $V = X \triangle \partial(B)$ where $B$ was obtained by adding any $2$-simplex with probability $0.5$. In the spaces not made using Vietoris-Rips, $X = \emptyset$. In the other examples $X$ was chosen to be the shortest $1$-cycle born at the same time as the most persistent $1$-cycle.

\subsection{Results and Discussion}

We tested our code on several different instances of the HL problem and then timed the algorithms on some of these problem instances to find out which was most efficient in practice. The results of our experiments are shown in \cref{FIG: Practical performance1} and \cref{FIG: Practical performance2}. 

The symbol $\infty$ is used to indicate that the program did not terminate on that input instance. This happened sometimes because we had set an upper limit to how much memory the computer could use. The algorithm based on the connectivity graph generally used more memory, something which is reflected in the number of times this algorithm was forced to shut down. There are several instances where neither algorithm terminated. This can not be seen in \cref{FIG: Practical performance1}, where these data points overlap, but can be spotted easily in \cref{FIG: Practical performance2}, where the value NaN indicates that neither algorithm terminated.

We  now explain why we believe that the data strongly suggests that the algorithm based on the Hasse diagram (algorithm H) is more efficient than the algorithm based on the connectivity graph (algorithm C). 

\begin{figure}[h!]
  \begin{subfigure}[b]{0.433\textwidth}
    \includegraphics[width=\textwidth]{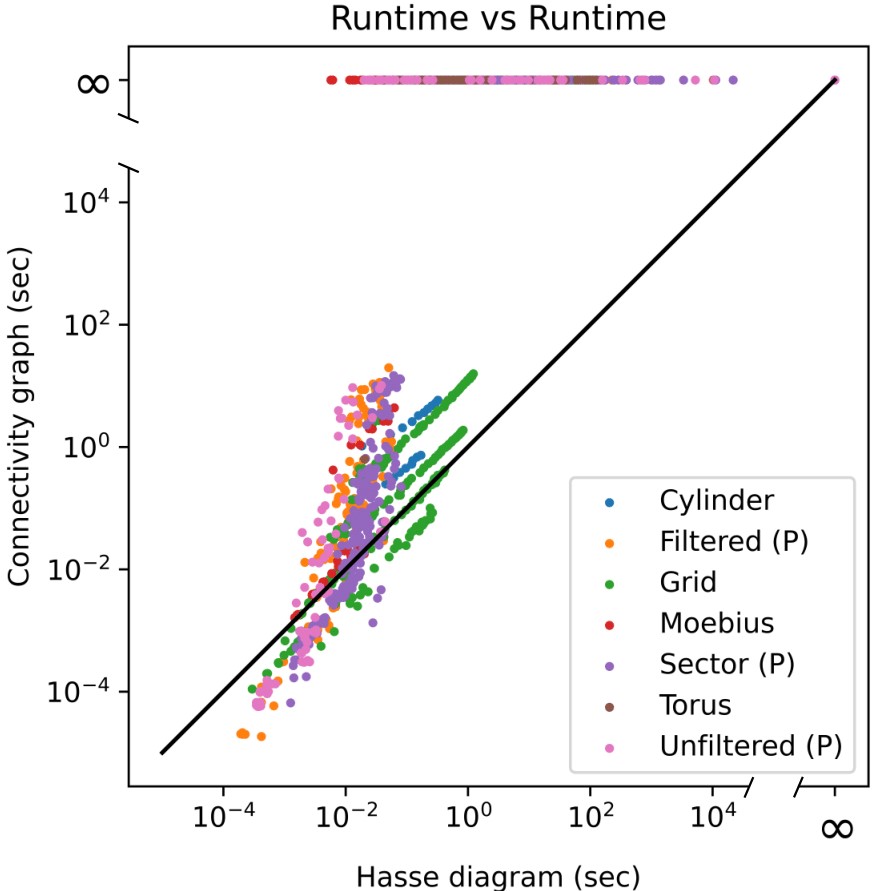}
    \label{fig: speed by speed}
  \end{subfigure}
  \begin{subfigure}[b]{0.55\textwidth}
    \includegraphics[width=\textwidth] {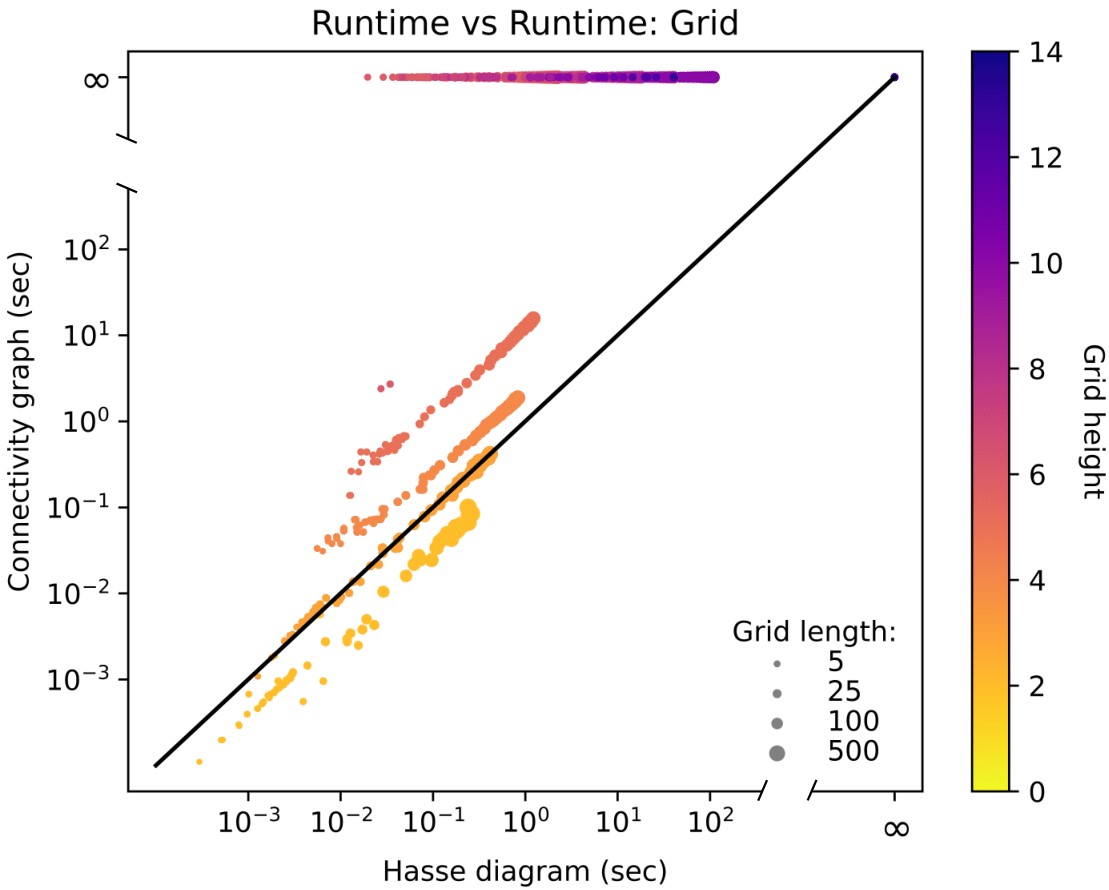}
    \label{fig: grid speed by speed}
  \end{subfigure}
  \caption{Each point represents an instance of the HL problem and its location tells us how fast it was solved. The plot to the left shows every problem instance. The kind of instance determine the colour of a given point. The rightmost plot only shows instances that are grids. The size/colour of a point is determined by the length/height of the grid.}
 
    \label{FIG: Practical performance1}
\end{figure}

In the plot to the left of \cref{FIG: Practical performance1} it is easy to see that algorithm H tends to do better than algorithm C. Still there are a significant number of cases where the opposite is true. A closer examination shows that this is not really a contradiction to our conclusion, since the only instances where algorithm C is faster are those where both algorithms finish very quickly. In particular, algorithm H is faster than algorithm C on every instance that takes more than a second to solve.

The plot to the right where we only plot the rectangular spaces gives some indication as to when algorithm C is likely to outperform algorithm H. There is a clear correlation between the height of the input grid and which algorithm is fastest while the length of the grid does not appear to have much of an impact. It shows that algorithm C is only better than algorithm H when the rectangle is very narrow and that the height of the grid has a big impact on how much better algorithm H does by comparison. Again, increasing the length does not have much of an impact. This indicates that algorithm C only outperforms algorithm H when the width of the tree decompositions is low i.e. this algorithm is the most effective only on spaces where the problem is easy. The next set of figures confirms this suspicion.

The decomposition of the underlying graph has a much bigger impact than the size of the input on a) how fast an instance is solved and on b) how much faster algorithm H solves the problem when compared with algorithm C. The leftmost plots in \cref{FIG: Practical performance2} gives evidence for the first part of this claim and the rightmost plots give evidence for the second. This fits well with with what we presented as the basic observation leading to parameterized algorithms: Problems are hard because they are complex, not because they are big. 

\begin{figure}[h!]
  \vspace{5mm}
  \begin{subfigure}[b]{0.5\textwidth}
    \includegraphics[width=\textwidth]{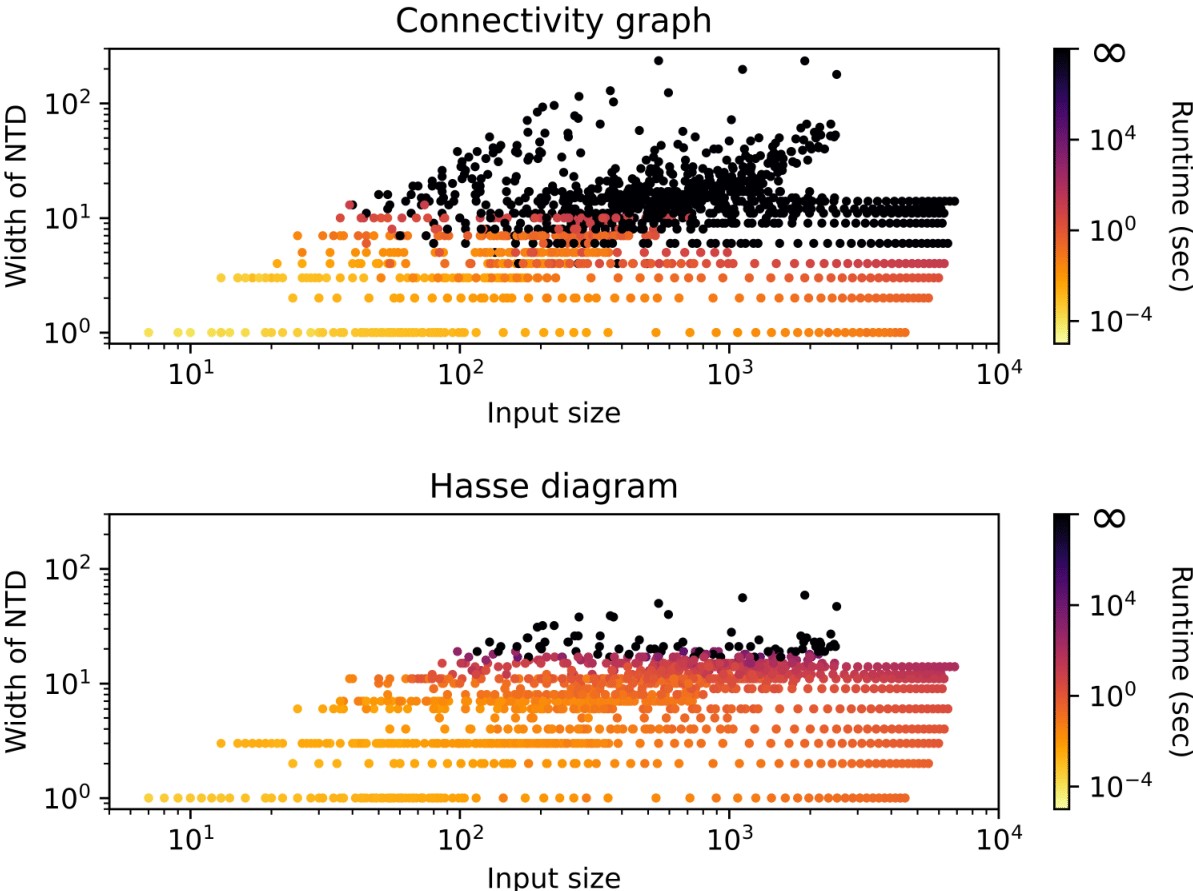}
    \label{fig: scatter width inputsize}
  \end{subfigure}
  \begin{subfigure}[b]{0.5\textwidth}
    \includegraphics[width=\textwidth] {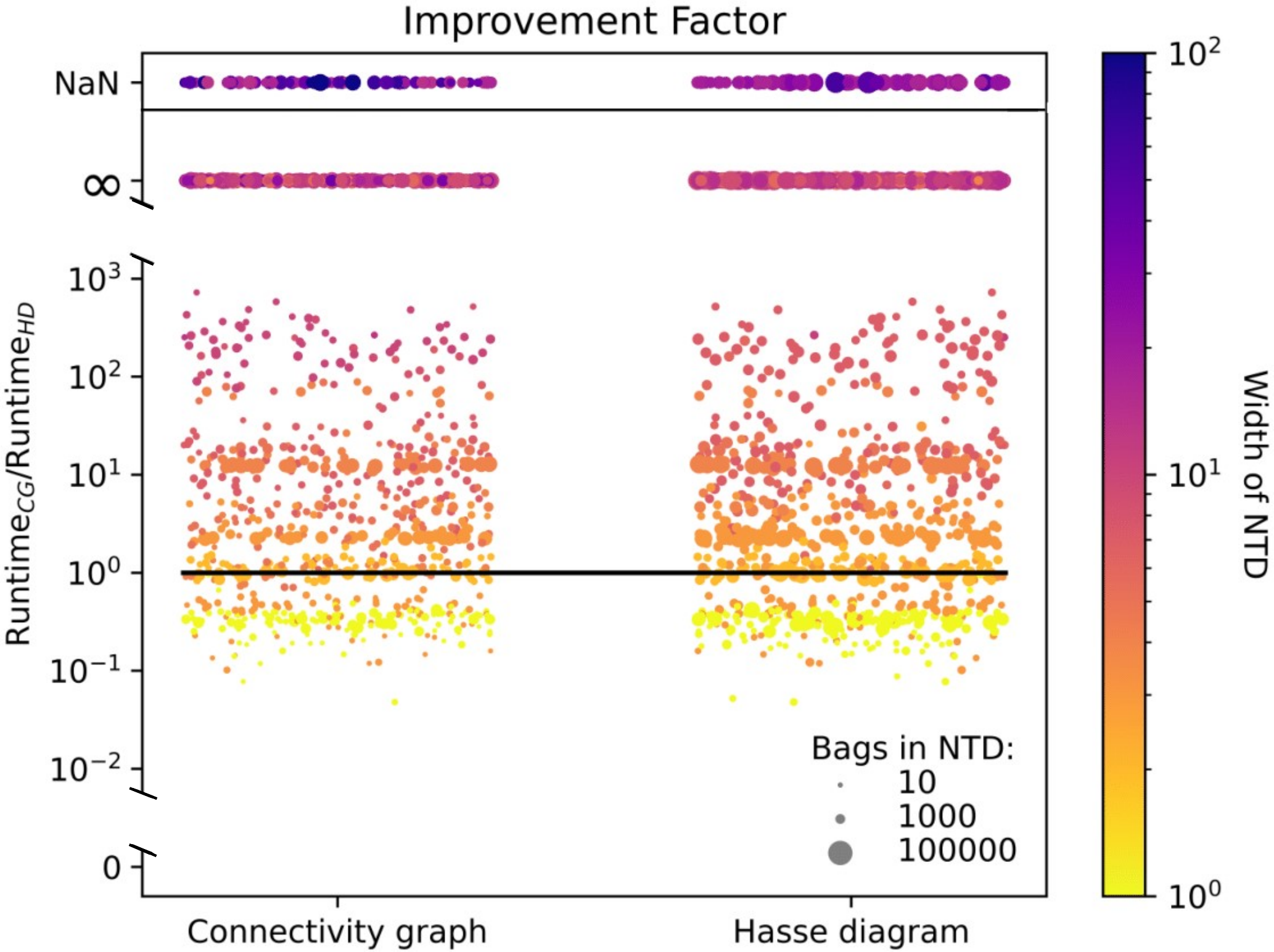}
    \label{fig: scatter inprovement ratio}
  \end{subfigure}
 
    \caption{The points are still instances but the rightmost and leftmost figure(s) now shows two plots, one for each algorithm. The position of a point in the leftmost figure is given by the width of the tree decomposition and the number of $1$ and $2$-simplices (combined). The colour of a point shows how fast it was solved. In the rightmost figure the $y$-axis gives the ratio of the runtime between the two algorithms. The size/colour of a point is given by the number of bags in/the treewidth of the NTD .}
    \label{FIG: Practical performance2}
\end{figure}

\subsection{Theoretical Explanation}

There are some good explanations for why algorithm H is faster when the treewidth is large. First, we have already seen in our runtime analysis that the dependency on the treewidth is substantially worse in the connectivity graph based algorithm. Further more we have the following proposition relating the treewidth of the two graphs to each other.

\begin{prop}
For any $K$ we have $tw(Hasse_{d}(K)) \leq tw(Con_{d}(K))+1 $.Conversely, $tw(Con_{d}(K))$ is not bounded by $tw(Hasse_{d}(K))$.
\end{prop}

\begin{proof}
Given a nice tree decompositon of $Con_{d}(K)$ we can make a tree decomposition of $Hasse_d(K)$ by potentially increasing the bag size by $1$. First, if the graph contains $d$-simplices with no cofaces then these have to be added separately. Otherwise, every $d+1$-coface of the $d$-simplex $\rho$ must occur simultaneously in a bag in $Con_{d}(K)$ as these form a complete subgraph (or clique). Add a copy of such a bag to the tree decomposition, attach it to the original bag and add $\rho$ to the copy. Repeating this process (only looking at the bags in the original tree decomposition) yields a tree decompositon of $Hasse_{d}(K)$.

For the last part simply note that there is a family of simplicial complexes $K^{(d,k)}$ where the treewidth of (level $d$ of) the Hasse diagram is $1$ and the treewidth of the connectivity graph is $k$. We construct this space by gluing $k$ individual $d+1$-simplices together along a common $d$-face. \cref{fig: comparison} illustrates the case of $d=2$ and $k= 8$. 

\begin{figure}[h!]
\centering
\hfill
  \begin{subfigure}[b]{0.25\textwidth}
  \centering
    \includegraphics[width=\textwidth]{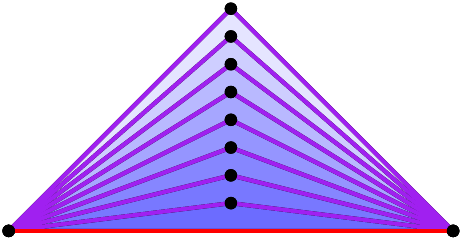}
  \end{subfigure}
\hfill
  \begin{subfigure}[b]{0.25\textwidth}
  \centering
    \includegraphics[width=\textwidth] {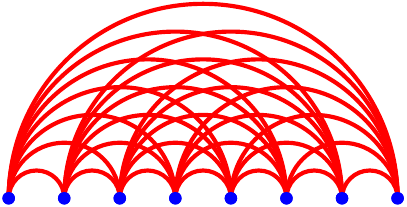}
  \end{subfigure}
\hfill
 \begin{subfigure}[b]{0.43\textwidth}
    \centering
    \includegraphics[width=\textwidth] {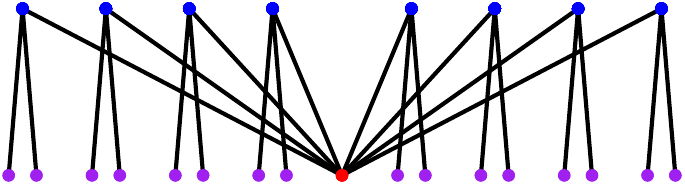}
  \end{subfigure}
  
\caption{\label{fig: comparison}The space $K^{(2,8)}$ with its corresponding $2$-connectivity graph and level $2$ of the Hasse diagram.}
  \end{figure}

\end{proof}

After all this, it may seem puzzling that algorithm C outperform algorithm H on any instances at all. As we already noted, this happens when the treewidth is low. This means its contribution to the runtime is negligible so the constant factor in the algorithm becomes dominating. The connectivity graph has fewer nodes than the Hasse diagram . This means that the size of the tree decomposition of the connectivity graph is also likely to be lower and thus fewer operations are needed. This goes at least some way in explaining why have these observations.

\section{Conclusion}
\label{sec:conclusion}
We view our results from a broader perspective and highlight some future research directions.

\subsection{The W[1]-Hardness Result}

We proved that the HL$_d$ problem is W[1]-hard for all $d \geq 1$ when parameterized by solution size. This is still the case when we restrict the input to spaces embeddable in $\mathbb{R}^{d+3}$ and also if we only ask for a constant factor approximation to an optimal solution. This followed easily from the work by Chen and Freedman in \cite{chen2011hardness}, but it is a new and important result for understanding the parameterized complexity of localizing homology.

In particular, this hardness result complements the FPT-algorithm in \cite{borradaile2020minimum} by Borradaile et. al. mentioned in the introduction. They solve the HL$_d$ on $d+1$ manifolds using solution size as a parameter. Assuming $\textbf{FPT} \neq W[1]$, our result shows that it is impossible to make a similar algorithm that works on general simplicial complexes.


\subsection{The FPT-algorithms}
We have designed and implemented two new FPT-algorithms for localizing homology in a simplicial complex $K$. Our experiments revealed that the algorithm using the treewidth $k$ of $Hasse_{d+1}(K)$ as a parameter was more efficient than the algorithm whose parameter is the treewidth of $Conn_{d+1}(K)$. These prototypes have already demonstrated practical potential. They were able to solve large instances as long as the treewidth of $Hasse_{d+1}(K)$ was low. Still, we believe that the efficiency of these algorithms can be further improved, allowing us to solve even larger instances.

We have several ideas for improvements on the technical front. Optimizing the code, rewriting it in a compiled programming language and using better algorithms for finding tree decompositions are three small steps to take that will likely lead to an even better performance. We could also look into designing preprocessing routines and try out various heuristics. Finally, if we want to get serious about making the fastest code possible, we should make use of parallel programming. Our algorithm does an enormous amount of very simple steps when computing solutions for large bags. These steps tend to be fairly independent from one another which should make this process ideal for parallelization. We would be very interested to see what the speedup would be if it was made to run in parallel. 

There is also more work to be done on the theoretical side of things. Recall that our best algorithm runs in $\mathcal{O}(4^k \cdot n)$-time and that under the ETH we can at most hope for a $\mathcal{O}(c^k \cdot n)$-time algorithm for some $c > 1$. 
We strongly suspect that the optimal value for $c$ is $2$ or higher. It is likely that a reduction similar to the one we gave in \cref{sec:eth} can be used to prove this under the strong exponential time hypothesis (SETH). This then begs the question: Is the base of $c= 4$ really the best we can do? At the time of writing we are agnostic about this. The join bag is the only step requiring $\mathcal{O}(4^k)$-time and so any significant improvement at this stage would give the entire algorithm a boost. 

\subsection{New Research Directions}

There are still many open questions surrounding homology localization, including 
\begin{itemize}
\itemsep0pt
    \item Which other useful parameterizations are there for the homology localization problem?
    \item Is there a form of homology localization where $d$-cycles can localized to outside of the $d$-skeleton of the simplicial complex ?
    \item Can we solve the homology localization problem on CW-complexes?
\end{itemize}

\section*{Acknowlegements}
The authors would like to thank both Michael Fellows and Lars Moberg Salbu for helpful comments on an earlier version of this manuscript. 

\bibliographystyle{plain}

\end{document}